\documentclass[journal,final,twocolumn,twoside]{IEEEtran}
\usepackage{amsmath,amsfonts}
\usepackage{algorithmic}
\usepackage{algorithm}
\usepackage{amssymb}
\usepackage{subfigure}
\usepackage{flushend}

\allowdisplaybreaks[4]
\usepackage{float} 

\usepackage{color}
\usepackage{array}
\usepackage[caption=false,font=normalsize,labelfont=sf,textfont=sf]{subfig}
\usepackage{textcomp}
\usepackage{stfloats}
\usepackage{url}
\usepackage{bm}
\usepackage{verbatim}
\usepackage{multirow} 
\usepackage{ntheorem}
\usepackage{graphicx}
\usepackage{cite}
\theorembodyfont{\upshape}
\theoremheaderfont{\it}
\qedsymbol{\square}
\theoremseparator{:}
\newtheorem{theorem}{\indent Theorem}
\newtheorem{lemma}{\indent Lemma}
\newtheorem*{proof}{\indent Proof}
\newtheorem{remark}{\indent Remark}

\newtheorem{proposition}{\indent Proposition}

\makeatletter

\newcommand{\Rmnum}[1]{\expandafter\@slowromancap\romannumeral #1@}
\makeatother
\begin{document}

\makeatletter
\let\myorg@bibitem\bibitem
\def\bibitem#1#2\par{%
  \@ifundefined{bibitem@#1}{%
    \myorg@bibitem{#1}#2\par
  }{%
    \begingroup
      \color{\csname bibitem@#1\endcsname}%
      \myorg@bibitem{#1}#2\par
    \endgroup
  }%
}
\makeatother 

\title{Energy-Efficient Edge Inference in Integrated Sensing, Communication, and\\Computation Networks}

\author{Jiacheng~Yao,~\IEEEmembership{Graduate Student Member,~IEEE},
        Wei Xu,~\IEEEmembership{Fellow,~IEEE}, 
        Guangxu~Zhu,~\IEEEmembership{Member,~IEEE},
        Kaibin~Huang,~\IEEEmembership{Fellow,~IEEE},
        and~Shuguang~Cui,~\IEEEmembership{Fellow,~IEEE}
\thanks{J. Yao and W. Xu are with the National Mobile Communications Research Laboratory, Southeast University, Nanjing 210096, China, and are also with the Purple Mountain Laboratories, Nanjing 211111, China (e-mail: \{jcyao, wxu\}@seu.edu.cn).}
\thanks{G. Zhu is with Shenzhen Research Institute of Big Data, Shenzhen, China (e-mail: gxzhu@sribd.cn).}
\thanks{K. Huang is with the Department of Electrical and Electronic Engineering, The University of Hong Kong, Hong Kong SAR, China (e-mail: huangkb@hku.hk).}
\thanks{S. Cui is with the School of Science and Engineering, The Chinese University of Hong Kong, Shenzhen, China (e-mail: shuguangcui@cuhk.edu.cn).}
}

%

\maketitle

\begin{abstract}
Task-oriented integrated sensing, communication, and computation (ISCC) is a key technology for achieving low-latency edge inference and enabling efficient implementation of artificial intelligence (AI) in industrial cyber-physical systems (ICPS). However, the constrained energy supply at edge devices has emerged as a critical bottleneck. In this paper, we propose a novel energy-efficient ISCC framework for AI inference at resource-constrained edge devices, where adjustable split inference, model pruning, and feature quantization are jointly designed to adapt to diverse task requirements. A joint resource allocation design problem for the proposed ISCC framework is formulated to minimize the energy consumption under stringent inference accuracy and latency constraints. To address the challenge of characterizing inference accuracy, we derive an explicit approximation for it by analyzing the impact of sensing, communication, and computation processes on the inference performance. Building upon the analytical results, we propose an iterative algorithm employing alternating optimization to solve the resource allocation problem. In each subproblem, the optimal solutions are available by respectively applying a golden section search method and checking the Karush-Kuhn-Tucker (KKT) conditions, thereby ensuring the convergence to a local optimum of the original problem. Numerical results demonstrate the effectiveness of the proposed ISCC design, showing a significant reduction in energy consumption of up to 40\% compared to existing methods, particularly in low-latency scenarios.
\end{abstract}

\begin{IEEEkeywords}
Edge inference, edge artificial intelligence (AI), industrial cyber-physical systems (ICPS), integrated sensing, communication and computation (ISCC). 
\end{IEEEkeywords}

\section{Introduction}
\IEEEPARstart{T}{he} rapid advancement of industrial cyber-physical systems (ICPS) has ushered in an era where the seamless integration of sensing, communication, and computation is paramount \cite{editor,icps1,icps0}. These systems form the backbone of modern smart industries, enabling real-time monitoring, control, and automation of industrial processes. However, the escalating complexity and stringent performance requirements of ICPS necessitate a paradigm shift in their design and implementation. There exists a fundamental conflict between limited processing capabilities and the increasing volume of real-time sensing data. Hence, artificial intelligence (AI) technology, renowned for its transformative capabilities, is expected to play a pivotal role in enhancing the capabilities of ICPS by enabling sophisticated data analysis \cite{xw}, predictive maintenance \cite{icps2}, anomaly detection \cite{editor1}, and autonomous decision-making \cite{icps3}.

Traditionally, deploying AI in ICPS has relied on powerful computational resources of centralized cloud servers. This approach involves transmitting vast amounts of raw data to the cloud for processing and analysis. While effective in leveraging advanced AI models, this method has inherent drawbacks, including high latency, significant bandwidth consumption, and potential privacy concerns. The need to transmit large volumes of data to central servers makes it difficult to meet the real-time and privacy requirements of industrial applications. To address these limitations, a new paradigm known as edge intelligence has emerged, which focuses on utilizing local data and computational resources at the edge of the network \cite{zhu2023pushing,wcm,intelledge,xu2023toward}. Edge intelligence aims to perform AI training and inference closer to the site where data is generated, thereby reducing latency, conserving bandwidth, and enhancing data privacy \cite{zhyang,gomore,imperfect,sciencechina}. In particular, edge inference, a key area in edge intelligence, concerns real-time inference and decision-making using a well-trained edge-deployed AI model, providing a platform for intelligent control in ICPS~\cite{edgeind}.

In the earliest edge inference architectures, edge devices typically serve a single role in either communication or computation. Inference tasks are performed either locally on edge devices or at central servers with preprocessed data from edge nodes. However, both approaches struggle to meet the demand for low-latency decision making in ICPS. Limited computing power on edge devices hampers fast inference of large-scale models, while constrained wireless resources restrict rapid data uploading to the server. These challenges are driving the emergence of split inference and further the joint design of communication and computation at the edge devices \cite{shao}. Specifically, large-scale AI models are split into two sub-models before being deployed respectively on edge devices and the servers. The sub-model at the edge extracts low-dimensional features from high-dimensional raw data, which can be transmitted to the server with limited bandwidth for completion of inference tasks. Within the framework of split inference, the communication and computation processes are tightly coupled, jointly determining the inference performance. Therefore, a joint design coordinated for both aspects becomes essential. To achieve the optimal trade-off between communication and computation, recent works have focused on various techniques, including adaptive splitting point selection \cite{flop,szbi}, model compression \cite{shao}, and setting early exiting \cite{exit0,exit}. Furthermore, taking into account the bottleneck of limited communication resource, advanced techniques, like joint source-channel coding \cite{shao2} and progressive feature transmission \cite{lanqiao}, are exploited to minimize communication overhead.

However, the aforementioned edge inference designs overlook the data collection process, despite that the quality of collected samples plays a cruial role in the resulting inference performance. In ICPS, seamless integration of sensing, communication, and computation collectively determines the efficacy of intelligent control. Therefore, conducting a joint design for integrated sensing, communication, and computation (ISCC) becomes imperative \cite{ISCCmag}. Benefiting from emerging technologies like integrated sensing and communication (ISAC), edge devices equipped with dual-functional hardware systems have ability to simultaneously sense, communicate, and compute \cite{fliu,zyhe}, which provides an essential foundation for the ISCC design. However, in ISAC, sensing and communication coexist as parallel tasks. In contrast, under ISCC, sensing, communication, and computation are integral procedures of inference tasks, not only coexisting functionally but also being interdependent in achieving the same task objective. Moreover, the absence of explicit objective expressions in inference tasks complicates the joint design of ISCC, making it more challenging than that in ISAC. 

Recently, there have been a few works considering the ISCC design for edge inference tasks \cite{gdyu,ISCC,AirISCC,zmzhuang}. For an action recognition task, the authors in \cite{gdyu} proposed a novel ISCC framework with sensing at the base station (BS) and AI inference at the edge server, and maximized the sensing performance with quality-of-service (QoS) requirements of tasks. For general classification tasks, the authors in \cite{ISCC} proposed a joint resource allocation scheme for ISCC directly targeting the inference accuracy measured by discriminant gain. This metric allows tractable derivation of the inference accuracy as a function of sensing, communication, and computation (SCC) resources, thus unveiling the impact of SCC resources on the inference performance. Furthermore, to fully exploit the multi-view features from multiple devices, over-the-air computation (AirComp) was employed for efficient feature fusion in \cite{AirISCC} and \cite{zmzhuang}, where joint resource allocation and beamforming design were further optimized in terms of the derived discriminant gain. From a theoretical perspective, the authors in \cite{mview} proposed an analytical framework for quantifying the fundamental performance gains brought by over-the-air aggregation of multi-view features. Moreover, within the extracted feature vector, the importance of each feature element was theoretically characterized in \cite{tvt}, which allows for more fine-grained resource allocation optimization adapting to element-wise feature importance.

Existing research primarily concentrates on enhancing inference performance through ISCC design under practical constraints like latency. However, it is essential to note that in real-world ICPS, battery-powered edge devices, e.g., internet-of-things (IoT) sensors or actuators, also encounter the challenging bottleneck of limited energy supply. In addition, in recent years, AI models have recently experienced exponential growth in size, significantly increasing the demand for data acquisition, model computation, and information exchange \cite{green}. This has exacerbated the dilemma of limited energy for edge devices, as they must manage not only the computational complexity of larger models but also  growing  requirements of communication and sensing. Hence, minimizing energy consumption at resource-limited edge devices under QoS requirements of inference tasks, e.g., accuracy and latency constraints, is a pressing need for the success of edge inference in practical ICPS.

As pointed out in \cite{green}, energy consumption in edge AI inference tasks is primarily determined by three key components: sensing, communication, and computation energy. Unfortunately, existing research on energy-efficient edge inference design often considered these components partially and separately, resulting in suboptimal performance. For instance, the authors in \cite{eeiscc} focused on resource allocation of sensing and communication for saving energy, while \cite{pimrc} addressed the joint design of communication and computation. This partial or separated design strategy overlooks the coupling relationship between sensing, communication, and computing. While it may achieve optimality at a single layer, it results in energy inefficiency from a holistic perspective. Therefore, joint design for sensing, communication, and computation is the key to improving energy efficiency, which is challenged by the following difficulties. Firstly, the performance of sensing, communication, and computation jointly determines inference accuracy, but the nonlinear and implicit coupling between these factors remains unexplored. In addition, the underlying mechanisms by which sensing, communication, and computation individually affect inference accuracy remain unclear. The lack of explicit characterization of these dimensions hinders the joint design and efficient resource allocation in ISCC systems, often resulting in reliance on empirical or heuristic approaches.
To the best of our knowledge, little effort has been made to elucidate the intrinsic mechanisms by which the three functional components affect inference accuracy, nor has a joint design approach been proposed to address the energy concerns for edge inference in ISCC networks.


Against this background, in this paper, we propose an energy-efficient ISCC framework as well as a corresponding resource allocation method for common classification tasks. The main contributions of our work are listed as follows.
\begin{itemize}
    \item \textbf{Energy-efficient ISCC framework}: 
    A flexible ISCC framework is established for energy-efficient edge inference. In concrete, we adopt the split inference method with flexible splitting points to meet the latency requirement and balance the communication and computational costs at the edge device. Depending on varying sensing qualities, we first propose to prune the device-side sub-model and employ stochastic quantization of features before transmission. Adjustable pruning ratio and quantization precision allow adaptive tuning of computation and communication capabilities to meet the requirement of inference accuracy and reduce redundant energy consumption. Considering the QoS constraints, we formulate an energy consumption minimization problem for resource allocation design within the proposed ISCC framework, which jointly optimizes the power allocation, selection of splitting point, pruning ratio, and the quantization precision.
    \item \textbf{Explicit inference accuracy characterization}: For the challenging task of characterizing inference accuracy in the formulated problem, taking classification tasks as examples, we derive a strict lower bound of the inference accuracy by analyzing the sufficient condition for a received feature sample can be correctly classified. To allow tractable optimization and insight extraction, we further derive an explicit and effective approximation of the derived bound of classification accuracy. It is revealed that sensing quality determines the achievable classification accuracy. On the other hand, the impacts from adaptive model pruning and stochastic feature quantization can be equivalently viewed as introducing additional additive noise on extracted features, thereby incurring performance loss.
    \item \textbf{Low-complexity resource allocation}: Building upon these analytical results, we propose an iterative algorithm for the ISCC resource allocation optimization. The original problem is divided into two subproblems, i.e., an accuracy-related subprobblem and a latency-related subproblem. For the first subproblem, we apply the typical golden section search method, which guarantees the global optimal solution. As for the second subproblem, we succeed to find the optimal solutions by checking the Karush-Kuhn-Tucker (KKT) conditions. Further by alternating iterations between the two subproblems, we rigorously ensure the convergence to at least a local optimum of the original problem.
    \item \textbf{Experiments}: We conduct extensive experiments to verify the performance analysis regarding classification accuracy and evaluate our proposed framework and resource allocation algorithm based on a typical human motion recognition task. Compared with existing baselines, the energy consumption obtained from the proposed method is greatly reduced, especially with stringent low-latency constraints. Moreover, it is shown that improving sensing qualities is of great importance to pursue high classification accuracy.
\end{itemize}

The rest of this paper is organized as follows. In Section~\Rmnum{2}, we offer the details of the proposed energy-efficient ISCC framework, and formulate the joint resource allocation optimization problem for energy consumption minimization. Section \Rmnum{3} characterizes the inference accuracy under the ISCC framework. In Section~\Rmnum{4}, we simplify the resource allocation problem and propose a low complexity algorithm. Simulation results and conclusions are given in Sections \Rmnum{5} and \Rmnum{6}, respectively.

Throughout the paper, numbers, vectors, and matrices are represented by lower-case, boldface lower-case, and boldface uppercase letters, respectively. The set $[N]$ is equal to $\{1,2,\cdots,N\}$. The operators $|\cdot|$ and $\left\|\cdot\right\|$ take the norm of a complex number and a vector, respectively. Moreover, the operator $|\cdot|$ returns the number of elements when the input is a set. The operator $\left\Vert \cdot\right\Vert_F$ denotes the Frobenius norm of a matrix. Let $\mathbb{R}$, $\mathbb{C}$, and $\mathbb{Z}^+$ denote the set of real numbers, complex numbers, and positive integers, respectively. 
The superscripts $(\cdot)^T$ and $(\cdot)^H$ stand for the transpose and conjugate-transpose operations, respectively. We use $\mathbb{E}[\cdot]$ to denote the expectation of a random variable (RV). 
The operator $\mathrm{Exp}(\lambda)$ is the exponential distribution with rate parameter $\lambda$.

\section{System Model}
As shown in Fig. \ref{fig:model}, we consider an edge inference system consisting of an ISCC edge device and a server. The edge device with limited resources aims at performing AI inference on real-time sensory data for classification tasks under the assistance of the server. The ISCC models and the proposed energy-efficient ISCC design are described in the following subsections.

\begin{figure}[!t]
  \centering
  \centerline{\includegraphics[width=3.2in]{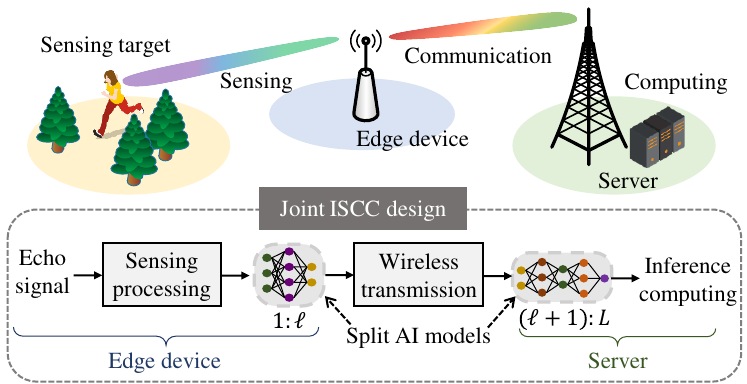}}
  \caption{Model of ISCC based edge inference system.}\label{fig:model}
\end{figure}

\subsection{Sensing Model}
\subsubsection{Sensing signal and processing}
We adopt the widely used frequency-modulated continuous wave (FMCW) signals with $M$ up-ramp chirps of duration $T_{0}$ for sensing. With the normalized FMCW signal $x(t)$ \cite{ISCC}, we write the echo signal as
\begin{align}\label{eq1}
    \!\!y(t)\!=\! \!\sqrt{P_{\text{S}}} h_0(t) x(t\!-\!\tau_0) \!+\!  \sqrt{P_{\text{S}}} \sum_{j=1}^J \! h_j(t) x(t\!-\!\tau_j)\!+\!n(t),
\end{align}
where $\sqrt{P_{\text{S}}}$ is the transmit power for sensing, $h_0(t)$ denotes direct attenuation of the target round-trip, $\tau_0$ is the round-trip delay, $h_j(t)$ and $\tau_j$, $j\in [J]$, denote the  reflection coefficient and signal delay of the $j$-th indirect path, respectively, and $n(t)$ is the Gaussian noise.

With the received echo signals, the edge device performs the following sensing processing to facilitate the subsequent AI inference. To begin with, the raw signal $y(t)$ is sampled with rate $f_s$, and then the sampled discrete sequence is reconstructed into a  two-dimensional sensing data matrix $\mathbf{Y}\in\mathbb{C}^{f_s T_0\times M}$. It is worth noting that the column dimension of $\mathbf{Y}$ is the fast-time dimension, while the row dimension is the slow-time dimension containing the feature in the Doppler spectrum shift. Then, according to \cite{peixi,ISCC}, we apply a  singular value decomposition (SVD)-based filter for $\mathbf{Y}$ to mitigate the clutter. Specifically, the filtered sensing data matrix is expressed as 
\begin{align}
    \bar{\mathbf{Y}}=\sum_{i=r_1}^{r_2} \varsigma_i \mathbf{u}_i \mathbf{v}_i^H,
\end{align}
where $\varsigma_i$, $\mathbf{u}_i$, and $\mathbf{v}_i$ represent the $i$-th singular value, the $i$-th left-singular vector, and the $i$-th right-singular vector of $\mathbf{Y}$, respectively, and $r_1\leq r_2$ are empirical parameters determined by experiments. Considering that only the slow-time dimension of $\bar{\mathbf{Y}}$ is useful for inference task, we convert it into a vector $\mathbf{\bar{y}}$ by summing up the data along the column dimension. Finally, short-time Fourier transform (STFT) is applied to generate a spectrogram $\mathbf{x}$, which serves as the input of AI model for inference task. To unify and standardize the learning and inference processes, all spectrograms $\mathbf{x}$ have been normalized to satisfy $\Vert \mathbf{x} \Vert=1$.

During the sensing procedure, the total latency and energy consumption are calculated as $T_{\text{sen}}=T_0 M$ and $E_{\text{sen}}=P_{\text{S}}T_{\text{sen}}$, respectively. Without loss of generality, we ignore the latency and energy consumption associated with the signal processing. 

\subsubsection{Sensing metrics}
According to (\ref{eq1}), the desired signal is interfered by higher-order scattering and noise. Consequently, the quality of the obtained spectrograms monotonically improves with the increase in transmit power 
$P_{\text{S}}$ until the noise becomes negligible. The experimental results in \cite{peixi} confirmed this conclusion by using the structural similarity (SSIM) index as the metric. Moreover, discriminant gain, derived from the well-established Kullback-Leibler (KL) divergence, is widely used as a substitute metric for inference accuracy \cite{lanqiao,ISCC}. Direct inspection of the discriminant gain in Eq. (17) of \cite{ISCC} reveals that it monotonically increases with sensing power, thereby establishing a positive relationship between the inference performance and sensing power.

\subsection{Split Inference Model}
Consider a pre-trained deep neural network (DNN) with $L$ layers for the AI inference. Due to limited computing resources and energy of the edge device, it is usually challenging to meet low latency requirements for the DNN calculations at the edge device. Additionally, it is often overloaded to support the direct upload of high-dimensional data sample to the server with limited wireless resources. To address these issues, we split the large-scale DNN at layer $\ell\in\mathcal{L}$, which is known as splitting point and $\mathcal{L}$ denotes its feasible set, and deploy the sub-models separately on the edge device and the server, as illustrated in Fig. \ref{fig:model}. Note that by selecting different splitting points, all possible situations are covered with this general splitting model. For example, when $\ell=L$, it represents the case that the AI inference is performed completely at the edge device.

To further reduce the computational burden on edge devices, in practice, the well-trained sub-models are usually pruned before deploying at the edge device via the existing model pruning methods \cite{imprun,weightprun,isik}. Let $\rho \in (0,1]$ denote the pruning ratio for the input sub-model at the edge, which represents the proportion of remaining nodes to the total number of nodes. Based on sensing quality and communication capability, we aim to adaptively adjust the inference ability of DNN by selecting a proper value of $\rho$, thereby achieving full utilization of the resources. We denote the number of floating point operations (FLOPs) required for computing the $l$-th layer by $\lambda(l,\rho)$. According to \cite{flop}, it follows
\begin{align}
    \lambda(l,\rho)\!=\!\left\{ 
    \begin{array}{ll}
    \left(2\gamma_{l-1}\psi_l^2\rho\!-\!1\right)\alpha_l \beta_l \gamma_l, & \text{Conv}, \\
    \alpha_l \beta_l \gamma_l \psi_l^2, & \text{MP},\\
    \left(2n_{l-1}\rho-1\right)n_l, & \text{FC},
    \end{array}\right.
\end{align}
where Conv, MP, and FC denote the convolutional, max-pooling, and fully-connected layers, respectively, $\alpha_l$, $\beta_l$, $\gamma_l$, and $\psi_l$ denote the height,
width, channels of the output feature map, and the filter size of the $l$-th layer, respectively, and $n_l$ denotes the number of neurons at the $l$-th layer. Then, the computation latency at the edge device and the server is respectively expressed as
\begin{align}
    T_{\text{comp},\text{e}}&=\frac{\sum_{l=1}^\ell\lambda(l,\rho)}{\nu_{\text{e}}},\nonumber \\
    T_{\text{comp},\text{s}}&=\frac{\sum_{l=\ell+1}^L\lambda(l,1)}{\nu_{\text{s}}},
\end{align}
where $\nu_\text{e}$ and $\nu_\text{s}$ denote the computation capacity of the edge device and server, respectively. They represent the number of FLOPs per second. According to \emph{Lemma 1} in \cite{kappa}, the energy consumption for computation at the edge device equals
\begin{align}
    E_{\text{comp}}=\kappa \sum_{l=1}^\ell \lambda(l,\rho) \nu_{\text{e}}^2, 
\end{align}
where $\kappa$ represents the effective switched capacitance depending on the chip architecture. In the considered edge inference system, the primary concern is the limited energy available at the edge device, while the central server typically has ample resources. Hence, the energy consumption at the server is not the focus of our study, and it is neglected for simplicity. Also, we assume a maximum CPU frequency for computation at the server to minimize computation latency and assist edge inference.

\begin{figure*}[!t]
  \centering
  \centerline{\includegraphics[width=6.2in]{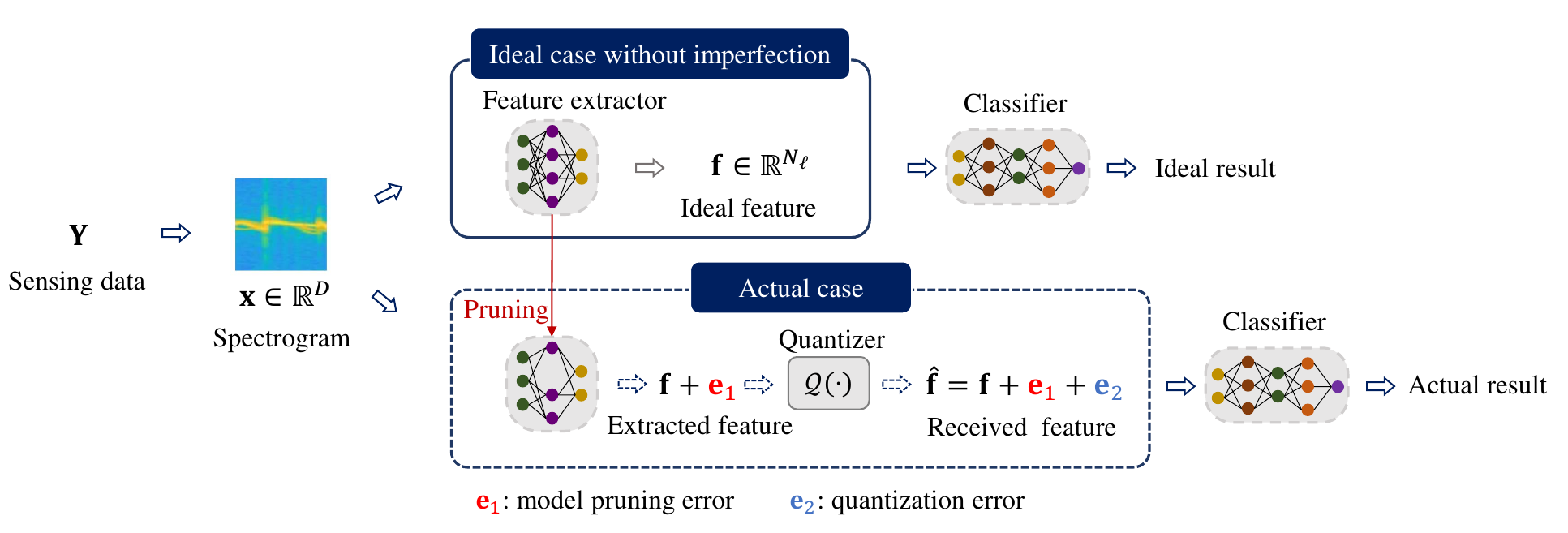}}
  \caption{Comparison between the classification results obtained from the proposed approach and the ideal case.}\label{fig:comparison}
\end{figure*}

\subsection{Communication Model}
Let $\mathbf{f}\in \mathbb{R}^{N_\ell}$ denote the output feature of the sub-model at the edge device with size $N_\ell$. In the communication procedure, the edge device uploads it to the server via a wireless link. Firstly, we adopt the typical stochastic quantization method \cite{quan,quanyao} to quantize the continuous feature elements into discrete ones for the convenience of transmission. Without loss of generality, we assume that $|f_i|\in \left[f_{\min}, f_{\max} \right]$, $\forall i\in [N_\ell]$. With $Q$ quantization bits, we divide the interval $\left[f_{\min}, f_{\max} \right]$ evenly into $2^{Q-1}-1$ quantization intervals and the $i$-th knob is denoted by $\tau_i=f_{\min}+\frac{f_{\max}-f_{\min}}{2^{Q-1}-1}i$ for $i=0,\cdots,2^{Q-1}-1$. Next, we define the quantization function as
\begin{align}\label{quan}
\mathcal{Q}(x)=\left\{ \begin{array}{cc}
\mathrm{sign}(x)\tau_i &\mathrm{w.p.} \enspace \frac{\tau_{i+1}-\vert x\vert}{\tau_{i+1}-\tau_i},\\
\mathrm{sign}(x)\tau_{i+1} &\mathrm{w.p.} \enspace \frac{\vert x\vert-\tau_i}{\tau_{i+1}-\tau_i},\\
\end{array} \right.
\end{align}
where $\vert x \vert \in [\tau_i,\tau_{i+1})$ and $\mathrm{sign}(\cdot)$ represents the signum function. Based on the quantization in (\ref{quan}), we obtain a quantized version of $\mathbf{f}$ for transmission, i.e., $\mathcal{Q}\left(\mathbf{f} \right)\triangleq \left[\mathcal{Q}\left(f_1\right),\cdots,\mathcal{Q}\left(f_{N_\ell} \right)\right]^T$. Moreover, the total number of bits needed for transmission equals $N_\ell Q$. 

After the quantization, the edge device transmits $\mathcal{Q}\left(\mathbf{f}\right)$ to the server via a wireless channel. We express the achievable rate as
\begin{align}
    r=B\log_2\left(1+\frac{g P_{\text{C}}}{B N_0} \right),
\end{align}
where $B$ is the available bandwidth, $g$ represents the channel gain between the edge device and server, $P_{\text{C}}$ is the transmit power, and $N_0$ denotes the noise power density. In order to guarantee lossless recovery of $\mathcal{Q}\left(\mathbf{f}\right)$ at the receiver, the minimum transmission latency must be
\begin{align}
    T_{\text{comm}}= \frac{N_\ell Q}{r}=\frac{N_\ell Q}{B\log_2\left(1+\frac{g P_{\text{C}}}{B N_0} \right)}.
\end{align}
Moreover, the energy consumption during the communication procedure amounts to
\begin{align}
    E_{\text{comm}}=P_{\text{C}}T_{\text{comm}}=\frac{N_\ell QP_{\text{C}}}{B\log_2\left(1+\frac{g P_{\text{C}}}{B N_0} \right)}.
\end{align}

\subsection{Problem Formulation}

With the considered edge inference task, sensing, communication, and computation are entirely coupled and the ISCC jointly determines the achievable performance. To meet the urgent demand of high-accuracy and low-latency edge inference, it is of great significance to explore an efficient resource management method at the resource-limited edge device and seek the optimal trade-offs among sensing, communication, and computation. The problem of resource allocation problem for the considered ISCC is formulated as follows:
\begin{align} \label{pro1}
\mathop{\text{minimize}}_{\ell,\rho,P_{\text{S}},P_{\text{C}},\nu_{\text{e}},Q} \enspace& E_{\text{sen}}+ E_{\text{comp}}+ E_{\text{comm}}\nonumber \\
\text{subject to}\enspace&\text{C}1:R_p\geq R_t,\nonumber \\
&\text{C}2: T_{\text{sen}}+ T_{\text{comp},\text{e}}+T_{\text{comp},\text{s}}+ T_{\text{comm}}\leq T_{\max},\nonumber \\
&\text{C}3: \ell \in \mathcal{L},\nonumber \\
&\text{C}4: \rho\in (0,1],\nonumber \\
&\text{C}5: P_{\text{S}},P_{\text{C}}\leq P_{\max},\nonumber \\
&\text{C}6: \nu_{\text{e}}\leq \nu_{\max},\nonumber \\
&\text{C}7: Q \in \mathbb{Z}^+, 
\end{align}
where $R_p$ denotes the inference accuracy, $R_t$ is a hyperparameter and denotes the target inference accuracy, $T_{\max}$ is the maximum permitted latency, and $P_{\max}$ and $\nu_{\max}$ denote the maximum transmit power and computation capacity of the edge device, respectively. 
The maximum permitted latency $T_{\text{max}}$ is defined as the target latency minus other potential delays, such as queuing delay. It serves as the upper bound for the total latency of the sensing, communication, and computation processes.
The problem in (\ref{pro1}) aims at minimizing the energy consumption required to achieve the inference accuracy and latency requirements under constraints $\text{C}1$ and $\text{C}2$,  through the joint optimization of splitting point, pruning ratio, power allocation, and the number of quantization bits. The remaining constraints, i.e., $\text{C}2-\text{C}7$, correspond to realistic constraints on these optimization variables. Moreover, in the considered single-user scenario, the entire bandwidth is allocated to the edge device for data upload. This allocation improves communication capability by increasing the available bandwidth, thereby enhancing inference performance without incurring additional energy consumption.

Through the formulated optimization problem in (\ref{pro1}), we aim to achieve the optimal trade-offs from the following two perspectives.

\emph{Latency trade-off:}
Given the assumption of fixed sensing time, we primarily consider the trade-off between communication and computation from the perspective of latency. For edge devices with limited resource, we optimize the splitting point to adaptively adjust the computational and communication burdens on the device, thus achieving the optimal latency.

\emph{Capability trade-off:} To achieve the desired inference accuracy, we strive to balance communication, sensing, and computation capabilities at the edge device, thereby minimizing the total energy consumption. Specifically, the sensing ability is adjusted through the sensing power, $P_{\text{S}}$. Computational capability, which refers to the AI model's ability to distinguish the input data samples, can be altered by adjusting the pruning rate $\rho$, with a larger $\rho$ indicating a stronger model capability. Additionally, communication capability is measured by the number of quantization bits, $Q$, with higher quantization accuracy leading to more precise data transmission.

\section{Inference Accuracy Characterization}\label{sec-III}
To pave the way for solving the problem in (\ref{pro1}), it is a prerequisite to provide an effective characterization of the inference accuracy with respect to the optimization variables in (\ref{pro1}). For the ease of exposition, we take classification tasks as an example in the subsequent analysis, by examining the essential sources of inference errors, we derive a strict lower bound on classification accuracy, as well as a closed-form approximation of the bound. The main notations functions used in this section are listed in Table \ref{notations}.

\begin{table}[t]
    \caption{Summary of Main Notations in Section \ref{sec-III}}\label{notations}
    \begin{center}
        \begin{tabular}{|m{1.4cm}<{\centering}|m{6.2cm}<{\centering}|}
            \hline
            \textbf{Notation} & \textbf{Definition}\\
            \hline
            $\mathbf{f}$ & Ideal feature vector without error\\ \hline
            $\hat{\mathbf{f}}$ & Received feature vector at the server\\ \hline
            $\mathbf{e}_1$ & Feature extraction error induced by model pruning\\ \hline
            $\mathbf{e}_2$ & Feature quantization error\\ \hline
            $\mathbf{W}_l$ & Weight matrix associated with the $l$-th layer\\ \hline
            $\hat{\mathbf{W}}_l$ & Pruned version of $\mathbf{W}_l$\\ \hline
            $N_l$ & Size of the output feature in the $l$-th layer \\ \hline
            $M_l$ & Number of parameters in $\mathbf{W}_l$\\ \hline
            $\delta(P_{\text{S}},\ell)$ & Classification margin under sensing power $P_{\text{S}}$ and splitting point $\ell$\\ \hline
            $R_0(P_{\text{S}})$ & Ideal classification accuracy with error-free feature and sensing power $P_{\text{S}}$\\ \hline
            $s$ & Lower bound for the minimum score of the obtained data samples\\
            \hline
        \end{tabular}
        \label{NOTATION}
    \end{center}
\end{table}

To begin with, we revisit the proposed approach in Fig.~\ref{fig:comparison}. Based on the split inference, we name the sub-models at the transmitter and receiver as feature extractor and classifier, respectively, according to their functions. Compared with the ideal case without model pruning and quantized transmission, the proposed approach suffers from two sources of errors, i.e., feature extraction error from the pruned model and quantization error. We respectively denote them as $\mathbf{e}_1$ and $\mathbf{e}_2$ and then denote the received feature vector at the receiver as $\hat{\mathbf{f}}\triangleq \mathbf{f}+\mathbf{e}_1+\mathbf{e}_2$, where $\mathbf{f}$ is the extracted feature in the ideal case.

Next, we explore the theoretical impact of the errors in the received features on classification accuracy from the perspective of the classifier. For effective characterization, we first introduce the concept of \emph{classification margin} \cite{cm}. As depicted in Fig. \ref{fig:cm}, the classification margin represents the infimum of the distance from an arbitrary feature vector, $\mathbf{f}$, to the classification boundary, which is related to the sensing quality and the learnt classifier. According to the definition of classification margin, a sufficient condition that the actual received feature $\hat{\mathbf{f}}$ remains to be correctly classified is
\begin{align} \label{necessary}
    \Vert \mathbf{f} -\hat{\mathbf{f}}  \Vert_2 \leq \delta(P_{\text{S}},\ell),
\end{align}
where $\delta(P_{\text{S}},\ell)$ denotes the classification margin under sensing power $P_{\text{S}}$ and splitting point $\ell$. Then, we are ready to derive an lower bound of the classification accuracy $R_p$ in the following lemma.

\begin{figure}[!t]
  \centering
  \centerline{\includegraphics[width=3in]{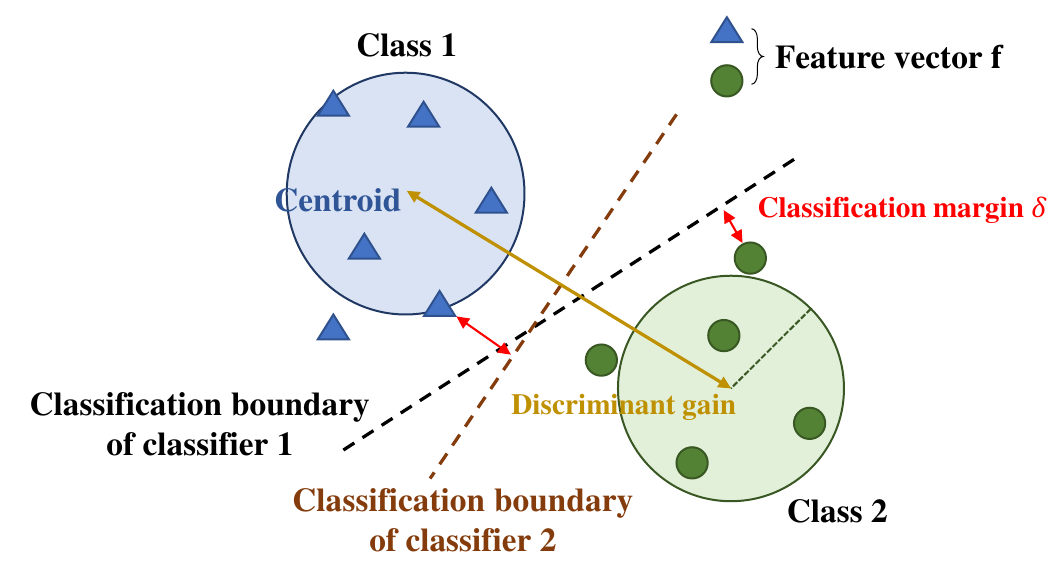}}
  \caption{Geometry of classification margin under a classification problem.}\label{fig:cm}
\end{figure}

\begin{lemma} \label{lemma1}
    The classification accuracy is lower bounded by
    \begin{align} \label{lower}
        R_p \geq R_0(P_{\text{S}}) \left(1- \frac{\mathbb{E}\left[\Vert \mathbf{e}_1\Vert^2 \right]+\mathbb{E}\left[\Vert \mathbf{e}_2\Vert^2 \right]}{\delta^2(P_{\text{S}},\ell)}\right),
    \end{align}
    where $R_0(P_{\text{S}})$ denotes the ideal classification accuracy with sensing power $P_{\text{S}}$, the expectation of $\Vert \mathbf{e}_1 \Vert^2$ is taken over random input data samples, and the expectation of $\Vert\mathbf{e}_2 \Vert^2$ is performed over the stochastic quantization error.
\end{lemma}

\begin{proof}
    According to the sufficient condition in (\ref{necessary}), we have 
     \begin{align}
        R_p &\geq R_0(P_{\text{S}}) \Pr \left\{ \Vert \mathbf{f}-\hat{\mathbf{f}}\Vert\leq \delta(P_{\text{S}},\ell)\right\}\nonumber \\
        &\overset{\text{(a)}}{\geq}R_0(P_{\text{S}}) \left(1- \frac{\mathbb{E}\left[\Vert \mathbf{e}_1+\mathbf{e}_2\Vert^2 \right]}{\delta^2(P_{\text{S}},\ell)}\right)\nonumber \\
        &\overset{\text{(b)}}{=} R_0(P_{\text{S}}) \left(1- \frac{\mathbb{E}\left[\Vert \mathbf{e}_1\Vert^2 \right]+\mathbb{E}\left[\Vert \mathbf{e}_2\Vert^2 \right]}{\delta^2(P_{\text{S}},\ell)}\right),
    \end{align}
    where the inequality in (a) comes from the Markov's inequality, and (b) is due to the independence between $\mathbf{e}_1$ and $\mathbf{e}_2$ and the fact that $\mathbb{E}[\mathbf{e}_2]=0$ \cite{quan}. This completes the proof. \hfill $\square$
\end{proof}

\begin{remark}
Eq. (\ref{lower}) reveals the fundamental mechanism by which the communication, computation, and sensing processes affect the classification accuracy. Imperfections in computation and communication can be represented as additive noise, primarily impacting the accuracy of features extracted from data samples. In contrast, sensing quality determines the achievable classification accuracy and the robustness against additional errors. Therefore, the core of the joint resource allocation problem in (\ref{pro1}) is to coordinate the sensing quality, computational error $\mathbf{e}_1$, and communication quantization error $\mathbf{e}_2$ under strict delay constraints to achieve desired accuracy.
\end{remark}

To proceeding, we need to further characterize the following terms, i.e., $R_0(P_{\text{S}})$, $\delta(P_{\text{S}},\ell)$, $\mathbb{E}\left[\Vert \mathbf{e}_1\Vert^2 \right]$ and $\mathbb{E}\left[\Vert \mathbf{e}_2\Vert^2 \right]$ in \emph{Lemma \ref{lemma1}}. The details of these terms are elaborated in the sequel.

\subsubsection{$R_0(P_{\text{S}})$} 
Considering the lack of solid theory on the relationship between sensing quality and classification accuracy, existing works mostly resort to qualitative analysis \cite{ISCC,AirISCC}. In this paper, we therefore conduct qualitative analysis on $R_0(P_{\text{S}})$ and determine the functional relationship through simulation testing. 

In specific, higher sensing power allows key task-related features to be more accurately represented in the data $\mathbf{x}$, thus improving inference accuracy. Conversely, under limited sensing power, the required echo signal becomes overwhelmed by noise, making it difficult for AI models to extract effective features for inference. Therefore, there is a positive correlation between the sensing power and classification accuracy.
To further characterize this implicit relationship, we used an empirical function as an alternative. Fig. \ref{acc_Ps} depicts the relationship between the actual classification accuracy and the sensing power, $P_{\text{S}}$. It is observed that the actual classification accuracy is well approximated by $R_{0}(P_{\text{S}})=a\cdot \mathrm{arctan}\left(bP_{\text{S}}\right)$, where $a>0$ and $b>0$ are parameters to be fitted.

\begin{figure}[!t]
  \centering
  \centerline{\includegraphics[width=3in]{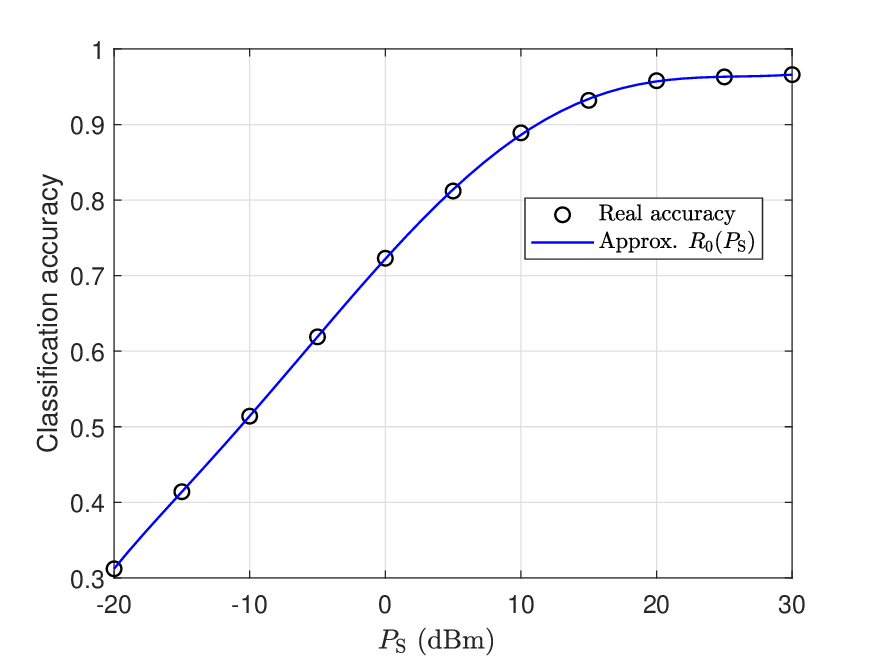}}
  \caption{Classification accuracy versus $P_{\text{S}}$.}\label{acc_Ps}
\end{figure}

\subsubsection{$\delta(P_{\text{S}},\ell)$}
We rely on the lower bound in \cite[Eq. (36)]{theorycm}, i.e.,
\begin{align}\label{deltal}
    \delta(P_{\text{S}},\ell)\geq\frac{\min_{\mathbf{x}}\{s(\mathbf{x})\}}{\prod_{i\in \mathcal{C}}\Vert \mathbf{W}_i \Vert_F}\triangleq \frac{s}{w(\ell)},
\end{align}
where $s(\mathbf{x})$ represents the score of a data sample $\mathbf{x}$, $\mathcal{C}\triangleq \{\ell+1,\cdots, L \}$, and $\mathbf{W}_i$ is the weight matrix associated with the $i$-th layer. For data sample with label $y$, $s(\mathbf{x})$ is defined~by 
\begin{align}
    s(\mathbf{x})\triangleq \min_{j\neq y} \sqrt{2}(\bm{\delta}_y -\bm{\delta}_j)^T f(\mathbf{x}),
\end{align}
where $f:\mathbb{R}^M\to \mathbb{R}^N$ denotes the trained neural network and $\bm{\delta}_j$ is the Kronecker delta vector with $[\bm{\delta}_j]_j=1$. The numerator in (\ref{deltal}) represents the minimum score of the obtained data samples, which is lower bounded by a constant, denoted by $s$. Although the minimum score of the obtained data samples is closely related to the value of $P_{\text{S}}$, for simplicity, we have adopted a universal lower bound as a substitute.

\subsubsection{Expected squared norm of $\mathbf{e}_1$}
Given the unknown distribution characteristics, we cannot directly calculate $\mathbb{E}\left[\Vert \mathbf{e}_1\Vert^2\right]$. Hence, we resort to find its upper bound in the following lemma as an alternative.
\begin{lemma}\label{lemma2}
    The square norm of error vector $\mathbf{e}_1$ is upper bounded by
    \begin{align}\label{upper}
    \Vert \mathbf{e}_1\Vert^2 \leq \sum_{l=1}^\ell \left(\! \left\Vert \mathbf{W}_l\!-\!\hat{\mathbf{W}}_l \right \Vert_F^2 \prod_{l^\prime =1,l^\prime\neq l}^\ell
    \left\Vert \mathbf{W}_{l^\prime}\right \Vert_F^2 \!\right),
    \end{align}
    where $\hat{\mathbf{W}}_l$ denotes the pruned version of $\mathbf{W}_l$.
\end{lemma}

\begin{proof}
    Please refer to Appendix \ref{app1}.\hfill $\square$
\end{proof}

By exploiting \emph{Lemma \ref{lemma2}}, we can straightforwardly derive a strict upper bound on $\mathbb{E}\left[\Vert \mathbf{e}_1\Vert^2 \right]$ under a given pruning strategy.
Note that this bound is a constant given the weights in the DNN, and it corresponding to the magnitude $\Vert \mathbf{e}_1\Vert^2$ in the worst-case scenario. However, the functional relationship between the pruning ratio $\rho$ and $\Vert \mathbf{e}_1\Vert^2$ is still perplexing and lacks a direct and concise form. This makes subsequent analysis and optimization challenging. To tackle this, we adopt a typical model pruning method as a representative, i.e., weight magnitude based pruning. In concrete, all parameters are sorted by their absolute value, and the parameters with the smallest $1-\rho$ fraction are set to 0. Given the specific pruning strategy, we derive a proper approximation in the following proposition.

\begin{proposition}\label{prop1}
    Under the weight magnitude based pruning, the upper bound in (\ref{upper}) can be approximated by a function of the pruning ratio, $\rho$, i.e.,
    \begin{align}\label{approx}
         \Vert \mathbf{e}_1\Vert^2 &\lesssim C(\ell) \underbrace{\left(2-\rho-\rho(\log\rho -1)^2\right)}_{\triangleq u(\rho)},
    \end{align}
    where $C(\ell)$ is defined by
    \begin{align}
        C(\ell)\triangleq \sum_{l=1}^\ell \left(\!\frac{M_l}{\lambda_l ^2}\prod_{l^\prime \neq l}^\ell
    \left\Vert \mathbf{W}_{l^\prime}\right \Vert_F^2 \!\right),
    \end{align}
    $M_l$ is the number of parameters in $\mathbf{W}_l$, and $\lambda_l$ is a constant introduced in Appendix \ref{app2}.
\end{proposition}

\begin{proof}
    Please refer to Appendix \ref{app2}.\hfill $\square$
\end{proof}

\subsubsection{Expected squared norm of $\mathbf{e}_2$}
According to \cite{quan}, with the stochastic quantization method, the expected squared norm of $\mathbf{e}_2$ is bounded by 
\begin{align}\label{quan2}
    \mathbb{E} \left [\Vert \mathbf{e}_2 \Vert^2\right]\leq \frac{\Delta(\ell)}{\left(2^{Q-1}-1\right)^2}\triangleq \Delta(\ell)v(Q),
\end{align}
where $\Delta(\ell)$ is 
\begin{align}
\Delta(\ell) \!\triangleq\! \left\{
\begin{array}{ll}
     \frac{1}{4}N_{\ell+1} (f_{\max}\!-\!f_{\min})^2&  \text{if } (\ell+1)\text{-th layer is MP},  \\
     \frac{1}{4}N_{\ell} (f_{\max}\!-\!f_{\min})^2& \text{otherwise}.
\end{array}\right.
\end{align}
It is important to note that if the layer following the splitting point is a maxpooling layer, the effective feature dimension impacting classification performance reduces. Consequently, $\Delta(\ell)$ depends on this actual feature dimension. In addition, due to the normalization operation in DNNs, the output features of each layer exhibit similar statistical characteristics. Therefore, it is reasonable to assume that $(f_{\max}-f_{\min})^2$ remains consistent across different layers, allowing us to establish a universal upper bound through, e.g., simulation tests.



Now, based on the obtained expressions in (\ref{lower}), (\ref{deltal}), (\ref{upper}), and (\ref{quan2}), we are ready to establish a strict lower bound, as well as an explicit and more tractable approximation, for the classification accuracy to optimize.

\section{Resource Allocation Optimization for ISCC}

\begin{figure}[!t]
  \centering
  \centerline{\includegraphics[width=2.8in]{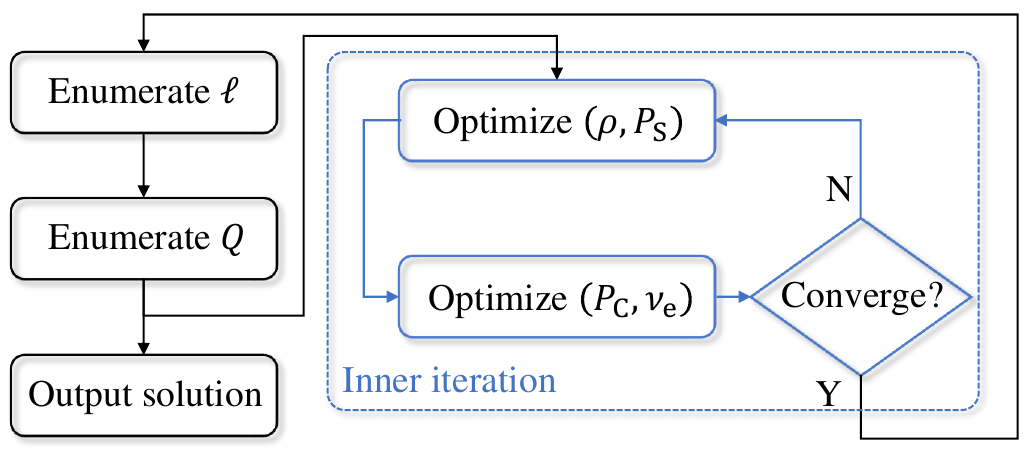}}
  \caption{Logical diagram of the proposed algorithm for solving problem (\ref{pro2}).}\label{fig:log}
\end{figure}

In this section, we propose an iterative algorithm with low complexity to solve the energy minimization problem, which is highly nonconvex. To begin with, building upon the derived expressions in Section \ref{sec-III}, we explicitly rewrite the classification accuracy constraint as 
\begin{align}\label{approxacc}
    \overline{\text{C}1}\!:R_0(P_{\text{S}})\!\left(1\!-\! \frac{w(\ell)}{s}\left(
    C(\ell)  u(\rho)\!+\! \Delta(\ell) v(Q)\right)\right)\!\geq\! R_t.
\end{align}
Then, the original problem in (\ref{pro1}) is reformulated as
\begin{align}\label{pro2}
\mathop{\text{minimize}}_{\ell,\rho,P_{\text{S}},P_{\text{C}},\nu_{\text{e}},Q} \enspace& E_{\text{sen}}+ E_{\text{comp}}+ E_{\text{comm}}\nonumber \\
\text{subject to}\enspace&\overline{\text{C}1},\text{C}2-\text{C}7.
\end{align}
Considering the challenge of tackling with the discrete constraints of splitting point, $\ell$, and the number of quantization bits, $Q$, we separate them from the optimization of other variables. For $\ell$, it represents the feasible split points in the DNN model. The depth of DNN models is typically limited, and not every layer is a feasible cut layer, so the optimization of $\ell$ only requires traversing a finite set $\mathcal{L}$. Similarly, for $Q$, we can set an upper limit $Q_{\max}$. When $Q\geq Q_{\max}$, it is safe to assume asymptotically error-free without quantization. Consequently, we only need to perform a finite search over the set $\{1,2,\cdots,Q_{\max}\}$ to determine the optimal $Q$. Therefore, a simple exhaustive search strategy suffices as an optimization approach for discrete variables $\ell$ and $Q$, which does not introduce a significant increase in complexity\footnote{For larger DNN models, more advanced enumeration methods, such as the branch-and-bound technique \cite{bnb}, must be explored to reduce computational complexity.}. This process is independent of the optimization of other variables and serves as the outer enumeration loop.

For the remaining coupled continuous variables, i.e., $(\rho,P_{\text{S}},P_{\text{C}},\nu_{\text{e}})$,
we optimize them using the alternating iteration method, which forms the inner iteration process shown in Fig. \ref{fig:log}. 
Given that the variables $(\rho, P_{\text{S}})$ are coupled in the classification accuracy constraint, $\overline{\text{C1}}$, in (\ref{approxacc}),  we provide a joint optimization for $(\rho, P_{\text{S}})$ in Sec. \ref{subsec1}. Moreover, $(P_{\text{C}},\nu_{\text{e}})$ are irrelevant to the constraint of classification accuracy and hence their optimizations are naturally separated in Sec. \ref{subsec2}. These two subproblems are alternately optimized to achieve a suboptimal solution for variables $(\rho,P_{\text{S}},P_{\text{C}},\nu_{\text{e}})$. The detailed algorithmic process is illustrated in Fig. \ref{fig:log}.

\subsection{Optimization over $(\rho, P_{\text{S}})$}\label{subsec1}
Given $(\ell,Q)$ and $(P_{\text{C}},\nu_{\text{e}})$, the problem in (\ref{pro2}) reduces to  
\begin{align}\label{pro4}
\mathop{\text{minimize}}_{\rho, P_{\text{S}}}  \enspace& P_{\text{S}}T_{\text{sen}}+\kappa \nu_{\text{e}}^2 \sum_{l=1}^\ell \lambda(l,\rho)\nonumber \\
\text{subject to}\enspace&\text{C}2:
\frac{\sum_{l=1}^\ell\lambda(l,\rho)}{\nu_{\text{e}}}\leq T_{1},\nonumber \\
&\overline{\text{C}1},\text{C}4,\text{C}5,
\end{align}
where $T_1 \triangleq T_{\max}-T_{\text{sen}}- T_{\text{comp},\text{s}}- T_{\text{comm}}$. Now, we are able to get the optimal $P_{\text{S}}$ in the following lemma.
\begin{lemma}
    If problem (\ref{pro4}) is feasible, the optimal solution, denoted by $P_{\text{S}}^*$, is
    \begin{align}\label{rho}
        P_{\text{S}}^*=r(u(\rho))
    \end{align}
    where $r(\cdot)$ represents the inverse function of $\frac{s}{w(\ell)C(\ell)} \left(1- \frac{R_t}{R_0(P_{\text{S}})}\right)-\frac{\Delta(\ell)v(Q)}{C(\ell)}$ with respect to $P_{\text{S}}$.
\end{lemma}
\begin{proof}
    Since the classification accuracy monotonically improves with growing $P_{\text{S}}$, the constraint $\overline{\text{C}1}$ can be reformulated as $P_{\text{S}}\geq r(u(\rho))$. Moreover, the objective function linearly increases with $P_{\text{S}}$. Hence, the objective is minimized with the minimum feasible sensing power. The proof completes. \hfill $\square$
\end{proof}

Substituting $P_{\text{S}}^*$ in (\ref{rho}) into problem (\ref{pro4}) and simplifying constraint $\text{C}2$, we have the problem equivalent to:
\begin{align} \label{pro7}
\mathop{\text{minimize}}_{\rho}  \enspace& T_{\text{sen}}r(u(\rho))+\kappa \nu_{\text{e}}^2 \sum_{l=1}^\ell \lambda(l,\rho)\nonumber \\
\text{subject to}\enspace&0<\rho\leq \rho_{\max},
\end{align}
where $\rho_{\max}$ is the unique solution to equation $\frac{\sum_{l=1}^\ell\lambda(l,\rho)}{\nu_{\text{e}}}= T_{1}$. Due to the implicit form of function $r(\cdot)$, the problem in (\ref{pro7}) can not be directly solved via  traditional gradient-based methods. To this end, we make the following assumption according to the sensing quality analysis in \cite{peixi}, which is further validated via the empirical analysis in Fig. \ref{acc_Ps}.

\emph{Assumption:} The derivative of $r(\cdot)$ is monotonically non-decreasing.

As $P_{\text{S}}$ increases, the quality of the generated spectrograms gradually saturates, leading to a diminishing improvement in the classification accuracy. This observation confirms the validity of this assumption. Now, we derive the following property of the objective function in problem (\ref{pro7}).

\begin{theorem}
    The objective function in problem (\ref{pro7}) is a unimodal function.
\end{theorem}

\begin{proof}
    By taking the derivative of the objective function and combining the \emph{Assumption}, it is straightforward to verify that it is monotonically increasing. Therefore, the derivative of the objective in (\ref{pro7}) has at most one zero point in the feasible domain, which completes the proof.
    \hfill $\square$
\end{proof}

Then, we adopt the golden section optimization method for solving (\ref{pro7}) due to its gradient-free characteristic. Also, for unimodal functions, it is well-known that the golden section search method always converges to the global optimal solution \cite{golden}. Hence, by applying the golden section search method, we can obtain the optimal $\rho$ to problem (\ref{pro7}). Detailed steps are summarized in Algorithm \ref{alg0}, where $h(\cdot)\triangleq T_{\text{sen}}r(u(\cdot))+\kappa \nu_{\text{e}}^2 \sum_{l=1}^\ell \lambda(l,\cdot)$.

\begin{algorithm}[!t]
\caption{Optimal solution to problem (\ref{pro4})} \label{alg0}
\renewcommand{\algorithmicrequire}{\textbf{Input:}}
\renewcommand{\algorithmicensure}{\textbf{Ouput:}}
\begin{algorithmic}[1] 
\REQUIRE $(\ell,Q,P_{\text{C}},\nu_{\text{e}})$, $T_{\text{sen}}$, $T_1$, $R_t$, $\left\{\mathbf{W}_l\right\}_{l=1}^L$, $\kappa$, $s$, $f_{\max}$, $f_{\min}$, and $R_0(\cdot)$
\ENSURE $(\rho^*,P_{\text{S}}^*)$
\STATE \textbf{Initialize} $\text{lb}=\rho_{\min}$, $\text{ub}=1$, $\varpi=\frac{-1+\sqrt{5}}{2}$ and convergence accuracy $\epsilon$.
\STATE Update $\tau_1=\text{lb}+(1-\varpi)(\text{ub}-\text{lb})$ and $\tau_2=\text{lb}+\varpi(\text{ub}-\text{lb})$.
\WHILE{$|\tau_1-\tau_2|>\epsilon$}
\IF{$h(\tau_1)<h(\tau_2)$}
\STATE Update $\text{ub}=\tau_2$.
\ELSE
\STATE Update $\text{lb}=\tau_1$.
\ENDIF
\STATE Update $\tau_1=\text{lb}+(1-\varpi)\text{ub}$ and $\tau_2=\text{ub}+(1-\varpi)\text{lb}$.
\ENDWHILE
\STATE Obtain $\rho^*=(\text{lb}+\text{ub})/2$ and update $P_{\text{S}}^*$ according to (\ref{rho}).
\RETURN $(\rho^*,P_{\text{S}}^*)$
\end{algorithmic} 
\end{algorithm}

\subsection{Optimization over $(P_{\text{C}},\nu_{\text{e}})$} \label{subsec2}
Given $(\ell,Q)$ and $(\rho,P_{\text{S}})$, we formulate the optimization with respect to $(P_{\text{C}},\nu_{\text{e}})$ as 
\begin{align}\label{pro5}
\mathop{\text{minimize}}_{P_{\text{C}},\nu_{\text{e}}} \enspace& \frac{A_1 P_{\text{C}}}{\log_2\left(1+\frac{g P_{\text{C}}}{B N_0} \right)}+\kappa A_2 \nu_{\text{e}}^2\nonumber \\
\text{subject to}\enspace& \frac{A_1}{\log_2\left(1+\frac{g P_{\text{C}}}{B N_0} \right)}+\frac{A_2}{\nu_{\text{e}}}\leq T_2,\nonumber\\
&\text{C}5,\text{C}6,
\end{align}
where $A_1\triangleq N_\ell Q/B$, $A_2\triangleq \sum_{l=1}^\ell \lambda(l,\rho)$, and $T_2\triangleq T_{\max}-T_{\text{comp},\text{s}}-T_{\text{sen}}$ are constants. Despite the nonconvexity of problem (\ref{pro5}), we are fortunately able to derive the optimal solution in closed form in the following lemma.
\begin{lemma}\label{lemma4}
    The optimal solution $(P_{\text{C}}^*,\nu_{\text{e}}^*)$ of the problem in (\ref{pro5}) satisfies
    \begin{align}\label{eq24}
        &P_{\text{C}}^* = \frac{BN_0}{g}\left(2^{\frac{1}{t^*}}-1\right),\nonumber \\
        &\nu_{\text{e}}^* = \min \left\{\nu_{\max},\left(\frac{\mu_1}{2\kappa}\right)^{1/3} \right\},
    \end{align}
    where $t^*$ is equal to
    \begin{align}\label{eq25}
        t^* = \max\left\{\frac{\mathrm{ln}2}{W\left( \frac{A_1}{e}\left(\mu_1-\frac{BN_0}{g}\right)\right)+1}, t_{\min}\right\},
    \end{align}
    $\mu_1$ is the solution to
    \begin{align}
        A_1 t^* +\frac{A_2}{\nu_e^*}=T_2,
    \end{align}
    $t_{\min}$ is a constant defined in Appendix \ref{app3}, and $W(\cdot)$ is the Lambert-$W$ function. The exact value of $\mu_1$ is obtained via a bisection search.
\end{lemma}

\begin{proof}
    Please refer to Appendix \ref{app3}.\hfill $\square$
\end{proof}

\subsection{Complexity and Convergence Analysis}

\begin{algorithm}[!t]
\caption{Proposed iterative algorithm for solving (\ref{pro2})} \label{alg1}
\renewcommand{\algorithmicrequire}{\textbf{Input:}}
\renewcommand{\algorithmicensure}{\textbf{Ouput:}}
\begin{algorithmic}[1]  
\REQUIRE $T_{\text{sen}}$, $T_{\max}$, $R_t$, $\left\{\mathbf{W}_l\right\}_{l=1}^L$, $B$, $g$, $N_0$, $\kappa$, $\nu_{\text{s}}$, $s$, $f_{\max}$, $f_{\min}$, and $R_0(\cdot)$
\ENSURE $(\ell,\rho,P_{\text{S}},P_{\text{C}},\nu_{\text{e}},Q)$
\FOR{$\ell\in \mathcal{L}$}
\FOR{$Q=1:Q_{\max}$}
\STATE \textbf{Initialize} iteration number $i=0$ and a sufficiently large $P_{\text{C}}^{(0)}$ and $\nu_{\text{e}}^{(0)}$to avoid infeasibility.
\REPEAT
\STATE Update $\left( \rho^{(i+1)},P_{\text{S}}^{(i+1)}\right)$ by solving problem (\ref{pro4}).
\STATE Update $(P_{\text{C}}^{(i+1)}, \nu_{\text{e}}^{(i+1)})$ according to \emph{Lemma \ref{lemma4}}.
\STATE Set $i=i+1$.
\UNTIL objective value in (\ref{pro2}) converges.
\ENDFOR
\ENDFOR
\RETURN the solution $(\ell,\rho,P_{\text{S}},P_{\text{C}},\nu_{\text{e}},Q)$ with the minimum energy consumption.
\end{algorithmic} 
\end{algorithm}

In summary, we conclude the main steps of solving  problem (\ref{pro2}) in Algorithm \ref{alg1}. We first analyze its computational complexity. For the first subproblem in (\ref{pro4}), the major complexity lies in the golden section search method, those complexity equals $\mathcal{O}\left(\log_2(1/\epsilon)\right)$. Besides, for the second subproblem in (\ref{pro5}), the complexity depends on the bisection method with complexity $\mathcal{O}\left(\log_2(1/\epsilon)\right)$. Thus, the total complexity of Algorithm \ref{alg1} is $\mathcal{O}\left(|\mathcal{L}| Q_{\max} I_{\text{itr}}\log_2(1/\epsilon)\right)$, where $I_{\text{itr}}$ denotes the required number of iterations for the alternative optimization.

As for the convergence, we note that the optimal solutions to the problems in (\ref{pro4}) and (\ref{pro5}) are attainable by the proposed algorithms, which yield to nonincreasing energy consumption in each step. Moreover, the total energy consumption is always lower bounded by a finite value, e.g., zero. Hence, we conclude that the proposed algorithm for optimizing $(\rho,P_{\text{S}}, P_{\text{C}},\nu_{\text{e}})$ always converges to a local optimum.

\begin{table}[!t]   
\renewcommand\arraystretch{1.1}
\begin{center}   
\caption{Simulation Parameters}  \label{table:1} 
\begin{tabular}{|l|l|l|} 
\hline \textbf{Notation} & \textbf{Interpretation} &\textbf{Value}\\
\hline  $R_t$&Target classification accuracy & 0.85\\
\hline  $T_{\max}$  &Maximum permitted latency & 0.8 s\\
\hline $\nu_{\max}$&Maximum CPU frequency of edge device  & 8e6 FLOP/s\\
\hline $\nu_{\text{s}}$&CPU frequency of the server & 1e11$\;$FLOP/s\\
\hline $\kappa$&The effective switched capacitance  &1e-21\\
\hline $T_0$&Chirp duration & 10 $\mu$s\\
\hline $f_s$&Sampling rate & 10 MHz\\
\hline $T_{\text{sen}}$&Unit sensing time & 0.5 s\\
\hline   $P_{\max}$&  Maximum transmit power of edge device &1 W \\ 
\hline    $B$&Available bandwidth for communication &0.1 MHz\\  
\hline   $\frac{g}{BN_0}$ &SNR for communication &20 dB\\
\hline 
\end{tabular}   
\end{center}  
\end{table}

\begin{figure}[!t]
  \centering
  \centerline{\includegraphics[width=3in]{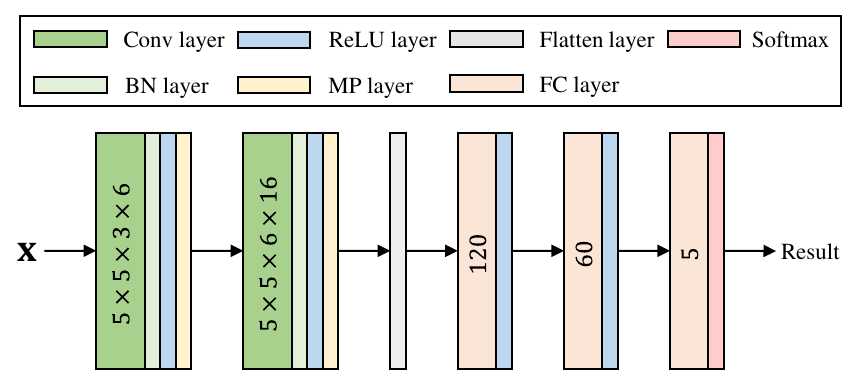}}
  \caption{Details of the adopted network.}\label{dnn}
\end{figure}

\section{Experimental Results}
In this section, numerical simulations are presented to validate the analytical results and the proposed resource allocation scheme.  We consider a typical human motion recognition task via ISCC at the edge device. Specifically, this task aims at distinguishing five human motions, i.e., standing, adult pacing, adult walking, child pacing, and child walking. We exploit the typical wireless sensing simulator in \cite{swang} to generate datasets for the DNN pre-training and data samples for edge inference. The differences in the five human motions primarily lie in height and movement speed. In particular, the heights of adults and children are uniformly distributed in the ranges [1.6m, 1.9m] and [0.9m, 1.2m], respectively. Additionally, the speeds for standing, walking, and pacing are set at 0 m/s, 0.5$H$ m/s, and 0.25$H$ m/s, respectively, where $H$ denotes the height. The heading of the moving human is uniformly distributed in the range $[-\pi,\pi]$.


The adopted AI model for inference task is a DNN with two
$5 \times 5$ convolution layers, three fully-connected layer with 120, 60, and 5 units, respectively,  and a softmax output layer. Max pooling operation is conducted following each convolutional layer and the activation function is ReLU. The details of the adopted DNN is shown in Fig. \ref{dnn}. The model is trained based on a training dataset with 12,000 data samples with no noise corruption and a testing dataset with 3,000 data samples. Unless otherwise specified, the other parameters are listed in Table \ref{table:1}.

For comparison, we consider the following baselines regarding the ISCC design and the corresponding resource allocation method.
\begin{itemize}
    \item On-server inference: The edge device upload all raw sensory data to the server and the inference task is fully accomplished at the server.
    \item Advanced on-device inference: The inference task is fully computed at the edge device and no communication procedure is involved. Also, model pruning strategy is conducted to improve energy efficiency.
    \item w/o pruning: Except for the no feature extractor pruning, everything else is consistent with the proposed method.
    \item Typical ISCC framework \cite{ISCC,AirISCC}: Principal component analysis (PCA) is adopted at the edge device for feature extraction and fixed classifier is deployed at the server.
    \item Joint Communication and Computation (JC\&C) design \cite{pimrc,vtc}: This scheme focuses solely on the joint optimization of communication and computation processes, including model splitting, model pruning, and adaptive feature quantization. The sensing process is not considered and is treated as a fixed input.
\end{itemize}
Among them, the first and second baselines serve as ablation experiments for the splitting selection in the proposed method, while the third baseline is an ablation experiment for model pruning. Additionally, by comparing with the typical ISCC framework, we aim to highlight the significance of optimization design for computational processes, such as adaptive model splitting and pruning. Meanwhile, compared to existing frameworks that focus solely on computation and communication, the importance of sensing design is emphasized.

\subsection{Classification Accuracy Approximation Performance}

\begin{figure}[!t]
  \centering
  \centerline{\includegraphics[width=3in]{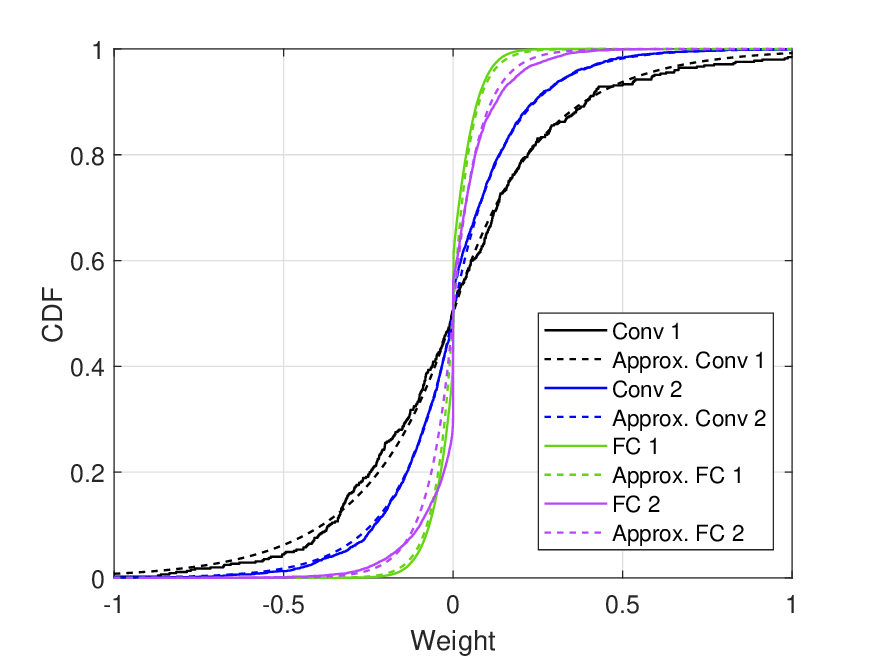}}
  \caption{Laplacian approximation for pre-trained model parameters.}\label{fig1}
\end{figure}

\begin{figure}[!t]
    \centering
    \subfigure[]{ \includegraphics[width=0.48\linewidth]{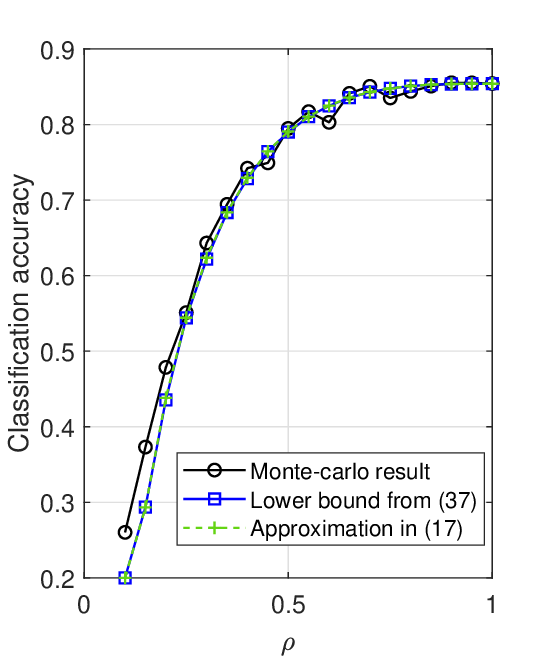}}
    \subfigure[]{	\includegraphics[width=0.48\linewidth]{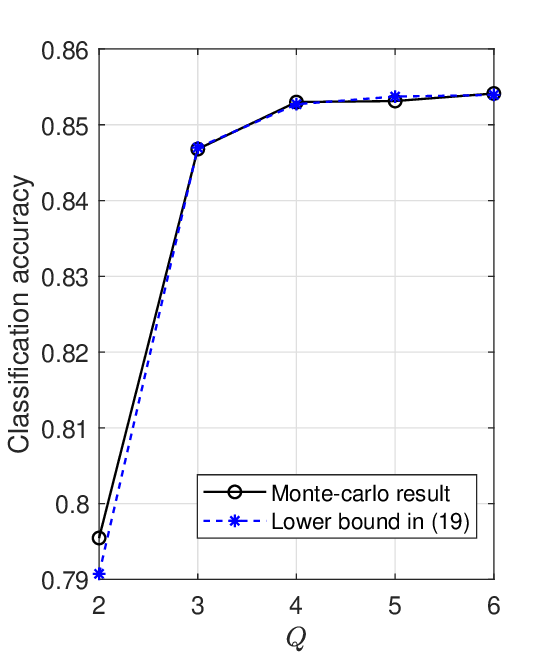}}
    \caption{\centering Classification accuracy versus (a) the pruning ratio, $\rho$, and (b) the number of quantization bits, $Q$.} \label{fig2}
\end{figure}


We validate the approximations in \emph{Proposition \ref{prop1}} in Fig.~\ref{fig1} and Fig. \ref{fig2}. Initially, we present the cumulative distribution function (CDF) of the model parameters across various layers. It is observed that the Laplace distribution employed in this study effectively approximates the parameter distribution in different layers of the pre-trained DNN, particularly in layers with a significant number of parameters, such as the second Conv layer and the first FC layer. This precise approximation of the weight distribution establishes a solid foundation for that subsequent theoretical analyses. Then, we verify the analysis regarding the pruning ratio, $\rho$. It is noteworthy that the lower bound we provide can effectively characterize the impact mechanism of various optimization variables on the classification accuracy and discern their trends. However, excessive scaling may lead to the final lower bound being much lower than the actual classification accuracy. To address this, we introduce a compensation constant for scaling the classification margin and it is determined numerically. As depicted in Fig.~\ref{fig2}(a), the derived approximation with scaling also aligns well with the actual curve obtained from Monte Carlo verification. It is evident that the parameters of the pre-trained DNN exhibit a degree of redundancy, allowing for pruning up to 50\% without significant performance loss in classification. This observation supports the pruning of DNNs to strike a balance between computational complexity and classification accuracy. Finally, we validate the analysis of the impact of quantization bits, $Q$, on classification in Fig. \ref{fig2}(b). It can be observed that the performance loss caused by quantization errors is negligible, especially compared to that by model pruning. Therefore, we are suggested to perform low-precision quantization on the extracted features to greatly reduce the communication overhead with marginal performance loss, thereby facilitating the energy-efficient ISCC design.

\subsection{Energy Consumption Performance}

\begin{figure}[!t]
  \centering
  \centerline{\includegraphics[width=3in]{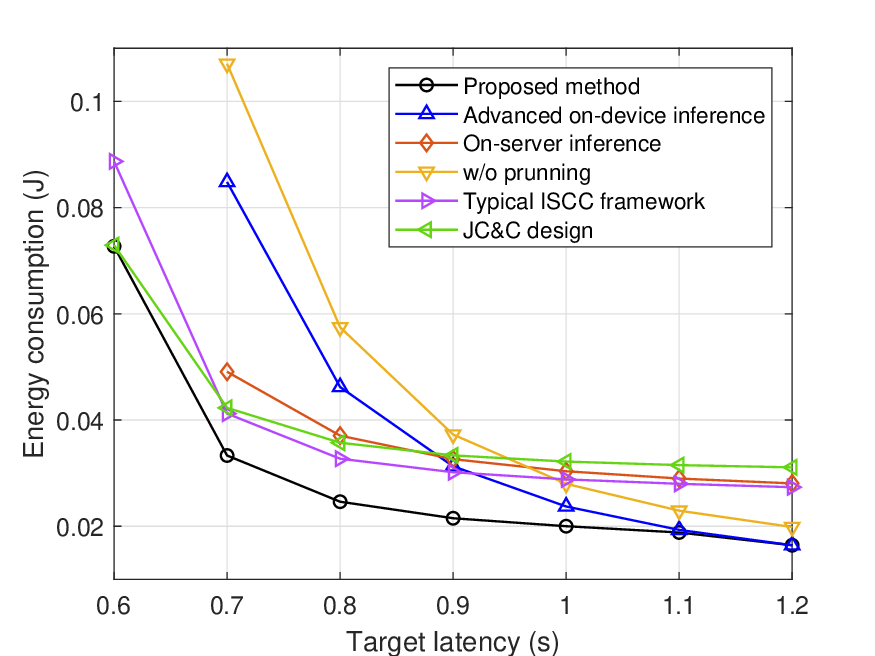}}
  \caption{Energy consumption versus the target latency, $T_{\max}$.}\label{latency}
\end{figure}

Fig. \ref{latency} depicts the total energy consumption versus the maximum tolerant latency, $T_{\max}$. It is observed that the proposed resource allocation method achieves the least energy consumption under diverse latency requirements.  When faced with extremely low latency requirements, such as 
$T_{\max}=0.6$ s, some baseline methods fail to complete inference tasks within the deadline. This is mainly due to the lack of adaptive adjustments for communication and computation overheads. In such cases, optimizing only the communication and computation processes in JC\&C scheme approaches the performance of the proposed method. This is because the proposed method does not optimize sensing time, necessitating an improvement in sensing quality to mitigate performance loss caused by incomplete communication and computation. Fixing a high sensing power is an energy-efficient choice in this case. Conversely, under more lenient latency constraint, executing the inference task locally is preferred. This is attributed to latency no longer being the primary limiting factor, thereby obviating the necessity for the server's powerful computing capabilities and avoiding the additional power consumption associated with the communication process. Notably, the JC\&C method exhibits the highest energy consumption in this scenario due to its lack of the consideration for sensing process.  The simulation results emphasize the necessity of jointly optimizing sensing, communication, and computation for energy-efficient inference. They further highlight the proposed scheme's advantages in adaptively adjusting resource allocation strategies based on task requirements

\begin{figure}[!t]
  \centering
  \centerline{\includegraphics[width=3in]{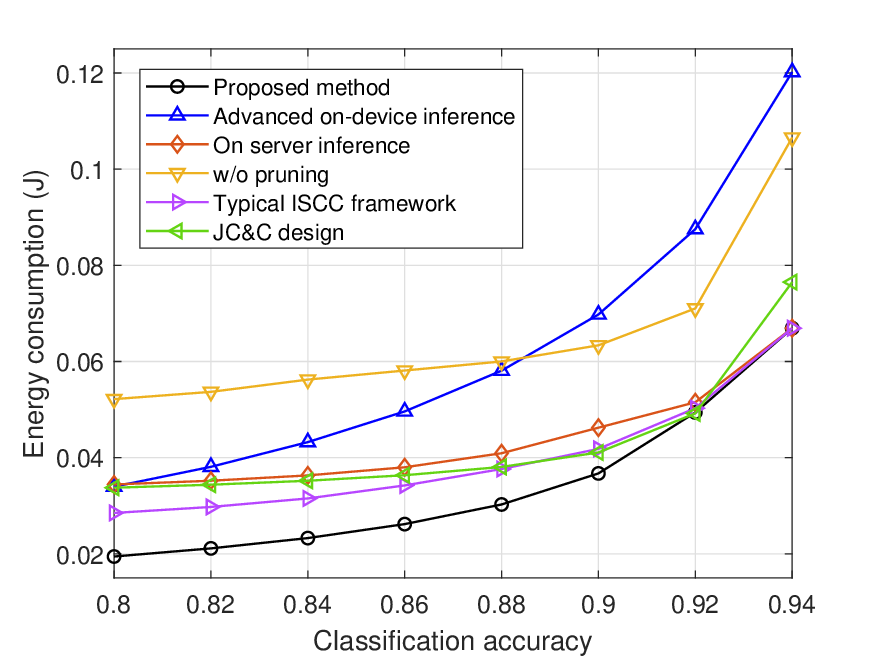}}
  \caption{Energy consumption versus the target classification accuracy, $R_t$.}\label{energy_acc}
\end{figure}

\begin{figure}[!t]
  \centering
  \centerline{\includegraphics[width=3in]{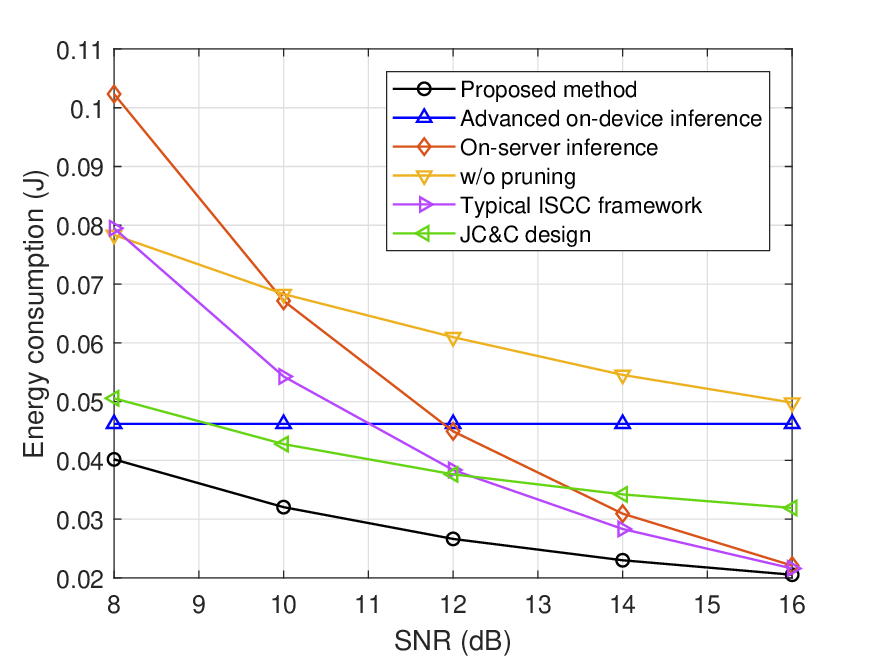}}
  \caption{Energy consumption versus SNR.}\label{snr}
\end{figure}

In Fig. \ref{energy_acc}, we show the comparison of energy consumption with target classification accuracy $R_t$. It is evident that achieving more accurate classification results in higher energy consumption. Meanwhile, our proposed algorithm also outperforms the other baselines. Besides, as the target accuracy increases, it becomes apparent that performing complete inference tasks at the edge device consumes the most energy. This is due to the fact that, compared to limited quantization errors, model pruning has a more pronounced effect on classification accuracy. To attain the desired classification accuracy, more resources must be allocated to compensate for the performance degradation caused by pruning errors. In addition, it is observed that under a high accuracy target, performing inference tasks at the server results in nearly the minimal energy consumption. This is because low-precision quantization can effectively reduce transmission latency without substantially compromising classification performance. Therefore, in scenarios with a strict latency constraint and high classification accuracy requirements, the optimal strategy typically involves allocating resources for sensing and communication processes.

Fig. \ref{snr} shows the impact of SNR on energy consumption, which further reflects the impact of communication capability. The strategy of completing all computation on the edge device does not involve communication process, so the energy consumption in this scenario remains invariant. Additionally, a higher SNR leads to a significant decrease in energy consumption for the other schemes. The substantial improvement in communication capability and efficiency due to increased SNR prompts the optimal strategy to upload more data to the server for further computation. This explains the asymptotic optimal performance of server-based inference under high SNR conditions. Additionally, it is worth noting that the typical ISCC framework in \cite{ISCC,AirISCC} can be seen as a special case of the on-server inference method, where additional dimensionality reduction is achieved through PCA. Hence, we observe that this method shows certain performance gains compared to the on-server inference approach, yet it still maintains a similar trend of change. By comparing with it, we further emphasize the necessity of optimizing for the computational process in the proposed method, especially in communication-constrained scenarios.

\begin{table*}[!t]
\renewcommand{\arraystretch}{1.2}
\caption{Comparison of Optimal parameter settings under different scenarios.}
\begin{center}
\begin{tabular}{|c|m{1.5cm}<{\centering}|m{1.5cm}<{\centering}|m{1.5cm}<{\centering}|m{1.5cm}<{\centering}|m{1.5cm}<{\centering}|m{2.3cm}<{\centering}|}
\hline
\multirow{2}{*}{\textbf{Scenarios}} &  \multicolumn{6}{c|}{\textbf{Optimal parameter settings}} \\ \cline{2-7} 
~ & {$\ell$} & {$Q$} & {$\rho$} & {$P_{\text{S}}$ (W)} & {$P_{\text{C}}$ (W)} & {$\nu_{\text{e}}$} (1e6 FLOP/s)\\ 
\hline
Latency-constrained & 5 & \textbf{5} & 0.3425 & 0.0189 & 0.0907 & \textbf{7.2871} \\
\hline
Accuracy-constrained  & 9 & 3 & \textbf{0.7409} & \textbf{0.0623} & 0.0169 & 2.6703 \\
\hline
SNR-constrained  & \textbf{11} & 2 & 0.6741 & 0.0145 & \textbf{0.1527} & 2.5698 \\
\hline
\end{tabular}
\label{trade-off}
\end{center}
\end{table*}

To further reveal the fundamental trade-offs, we show the optimal parameter settings obtained via the proposed resource allocation method under different scenarios in Table \ref{trade-off}. In particular, we consider three typical scenarios, i.e., latency-constrained $(R_t,T_{\max},\text{SNR})=(0.8,0.6,100)$, accuracy-constrained $(R_t,T_{\max},\text{SNR})=(0.94,1.2,100)$, and SNR-constrained $(R_t,T_{\max},\text{SNR})=(0.8,1.2,1)$ scenarios. To achieve low-latency requirements, extensive reduction of redundant parameters is crucial as it significantly enhances computational speed. In scenarios requiring high classification accuracy, augmenting sensing capabilities and improving the quality of data samples become paramount. This is because the quality of sensing sets the upper limit for achievable classification accuracy. In scenarios constrained by communication resources, maximizing computation at the edge device is essential to achieve data uploads minimization. Moreover, as the split point moves further along the network, the required quantization accuracy decreases. This is because, after multiple layers of computation, the extracted features become increasingly distinct and robust to quantization errors, resulting in minimal loss in inference accuracy even with lower quantization precision.

\section{Conclusion}
In this paper, we considered the energy-efficient ISCC design for edge inference. A joint resource allocation problem, including the optimization of power allocation, splitting design, pruning ratio and quantization precision, was formulated to minimize the energy consumption at the edge device with classification accuracy and latency requirements. The simulation results verified the effectiveness of the proposed ISCC design and confirmed its superiority over baseline approaches, especially with stringent latency constraint. Moreover, it is suggested that improving sensing qualities is preferred to enhance classification accuracy. Further, extending the proposed approach to multi-device collaborative inference scenarios is a promising direction for future work.

\appendices

\section{Proof of Lemma \ref{lemma2}} \label{app1}
As a preparation, we first represent the trained feature extractor as follows
\begin{align}\label{a0}
    f_e(\mathbf{x})=\mathbf{W}_\ell \sigma\left( \mathbf{W}_{\ell-1}\sigma\left( \cdots \mathbf{W}_{2}\sigma\left( \mathbf{W}_{1}\mathbf{x}\right)\right)\right)
\end{align}
where $\sigma$ denotes the ReLU activation. Moreover, we can also represent the pruned feature extractor in a similar form
\begin{align}
    \hat{f}_e(\mathbf{x})=\hat{\mathbf{W}}_\ell \sigma\left( \hat{\mathbf{W}}_{\ell-1}\sigma\left( \cdots \hat{\mathbf{W}}_{2}\sigma\left( \hat{\mathbf{W}}_{1}\mathbf{x}\right)\right)\right).
\end{align}
Now, we can reformulate the square norm of $\mathbf{e}_1$ in (\ref{apa0}) on the top of this page, where (a) is due to the triangle inequality and the Cauchy-Schwarz inequality, (b) is due to the 1-Lipschitzness of ReLU activation with respect to $\ell_2$ norm. 
\begin{figure*}
\begin{align}\label{apa0}
    \Vert \mathbf{e}_1 \Vert^2 &=\Vert f_e(\mathbf{x})-\hat{f}_e(\mathbf{x}) \Vert^2=\left \Vert \mathbf{W}_\ell \sigma\left( \mathbf{W}_{\ell-1}\sigma\left( \cdots \mathbf{W}_{2}\sigma\left( \mathbf{W}_{1}\mathbf{x}\right)\right)\right)-\hat{\mathbf{W}}_\ell \sigma\left( \hat{\mathbf{W}}_{\ell-1}\sigma\left( \cdots \hat{\mathbf{W}}_{2}\sigma\left( \hat{\mathbf{W}}_{1}\mathbf{x}\right)\right)\right)
    \right\Vert^2\nonumber \\
    & = \left \Vert \mathbf{W}_\ell \sigma\left( \mathbf{W}_{\ell-1}\sigma\left( \cdots \mathbf{W}_{2}\sigma\left( \mathbf{W}_{1}\mathbf{x}\right)\right)\right)-
    \hat{\mathbf{W}}_\ell \sigma\left( \mathbf{W}_{\ell-1}\sigma\left( \cdots \mathbf{W}_{2}\sigma\left( \mathbf{W}_{1}\mathbf{x}\right)\right)\right)\right.\nonumber \\
    &\left. \quad \quad +\hat{\mathbf{W}}_\ell \sigma\left( \mathbf{W}_{\ell-1}\sigma\left( \cdots \mathbf{W}_{2}\sigma\left( \mathbf{W}_{1}\mathbf{x}\right)\right)\right)-
    \hat{\mathbf{W}}_\ell \sigma\left( \hat{\mathbf{W}}_{\ell-1}\sigma\left( \cdots \hat{\mathbf{W}}_{2}\sigma\left( \hat{\mathbf{W}}_{1}\mathbf{x}\right)\right)\right)
    \right \Vert^2 \nonumber \\
    &\overset{\text{(a)}}{\leq} \left \Vert \mathbf{W}_\ell \!-\! \hat{\mathbf{W}}_\ell\right \Vert_F^2 \left \Vert \sigma\left( \mathbf{W}_{\ell-1}\sigma\left( \cdots \mathbf{W}_{2}\sigma\left( \mathbf{W}_{1}\mathbf{x}\right)\right)\right)\right\Vert^2 \nonumber\\
    &\quad +\Vert \hat{\mathbf{W}}_\ell \Vert_F^2 \left \Vert \sigma\left( \mathbf{W}_{\ell-1}\sigma\left( \cdots \mathbf{W}_{2}\sigma\left( \mathbf{W}_{1}\mathbf{x}\right)\right)\right)\!-\!
   \sigma\left( \hat{\mathbf{W}}_{\ell-1}\sigma\left( \cdots \hat{\mathbf{W}}_{2}\sigma\left( \hat{\mathbf{W}}_{1}\mathbf{x}\right)\right)\right)
    \right \Vert^2 \nonumber \\
    &\overset{\text{(b)}}{\leq}\left \Vert \mathbf{W}_\ell \!-\! \hat{\mathbf{W}}_\ell\right \Vert_F^2 \left \Vert \mathbf{W}_{\ell-1}\sigma\left( \cdots \mathbf{W}_{2}\sigma\left( \mathbf{W}_{1}\mathbf{x}\right)\right)\right\Vert^2\! +\!\Vert \hat{\mathbf{W}}_\ell \Vert_F^2 \left \Vert \mathbf{W}_{\ell-1}\sigma\left( \cdots \mathbf{W}_{2}\sigma\left( \mathbf{W}_{1}\mathbf{x}\right)\right)\!-\!
    \hat{\mathbf{W}}_{\ell-1}\sigma\left( \cdots \hat{\mathbf{W}}_{2}\sigma\left( \hat{\mathbf{W}}_{1}\mathbf{x}\right)\right)
    \right \Vert^2.
\end{align}
\hrule
\end{figure*}
For the first term on the right hand side (RHS) of (\ref{apa0}), by exploiting the fact that $\Vert\mathbf{Wx}\Vert \leq \Vert \mathbf{W}\Vert_F \Vert \mathbf{x} \Vert$, we have 
\begin{align}\label{apa1}
    &\left \Vert \mathbf{W}_\ell \!-\! \hat{\mathbf{W}}_\ell\right \Vert_F^2 \left \Vert \mathbf{W}_{\ell-1}\sigma\left( \cdots \mathbf{W}_{2}\sigma\left( \mathbf{W}_{1}\mathbf{x}\right)\right)\right\Vert^2\nonumber \\
    &\quad \leq \left \Vert \mathbf{W}_\ell \!-\! \hat{\mathbf{W}}_\ell\right \Vert_F^2\left( \prod_{l=1}^{\ell-1} \left \Vert \mathbf{W}_l\right \Vert_F^2 \right)\Vert \mathbf{x}\Vert^2. 
\end{align}
Then, following the same procedures in (\ref{apa0}) and (\ref{apa1}), $\Vert \mathbf{e}_1 \Vert^2$ is bounded by (\ref{apa2}) on the top of the next page, where (a) is due to the fact that $\Vert \hat{\mathbf{W}}_l\Vert^2 \leq \Vert \mathbf{W}_l\Vert^2$. The proof completes.
\begin{figure*}
\begin{align}\label{apa2}
    \Vert \mathbf{e}_1 \Vert^2 &\overset{\text{(a)}}{\leq} \left \Vert \mathbf{W}_\ell \!-\! \hat{\mathbf{W}}_\ell\right \Vert_F^2\left( \prod_{l=1}^{\ell-1} \left \Vert \mathbf{W}_l\right \Vert_F^2 \right)\Vert \mathbf{x}\Vert^2\! +\!\Vert \mathbf{W}_\ell \Vert_F^2 \left \Vert \mathbf{W}_{\ell-1}\sigma\left( \cdots \mathbf{W}_{2}\sigma\left( \mathbf{W}_{1}\mathbf{x}\right)\right)\!-\!
    \hat{\mathbf{W}}_{\ell-1}\sigma\left( \cdots \hat{\mathbf{W}}_{2}\sigma\left( \hat{\mathbf{W}}_{1}\mathbf{x}\right)\right)
    \right \Vert^2\nonumber \\
    &\leq \sum_{l=1}^\ell \left \Vert \mathbf{W}_l \!-\! \hat{\mathbf{W}}_l\right \Vert_F^2\left( \prod_{l^\prime \neq l} \left \Vert \mathbf{W}_l\right \Vert_F^2 \right)\Vert \mathbf{x}\Vert^2\! 
\end{align}
\hrule
\end{figure*}


\section{Proof of Proposition \ref{prop1}} \label{app2}
Since $\rho$ only affects the first term of the upper bound in (\ref{upper}), $\left\Vert \mathbf{W}_l\!-\!\hat{\mathbf{W}}_l \right \Vert_F^2$, and it does not impact the second term, we retain the second term while approximating only the first term for improved precision. By defining $d_l=\left\Vert \mathbf{W}_l-\hat{\mathbf{W}}_l \right \Vert_F^2$ and combining the adopted pruning strategy, we have 
\begin{align}\label{apb1}
    \mathbb{E}\left[ d_l\right] = \mathbb{E}\left[ \sum_{k=1}^{\lfloor(1-\rho)M_l \rfloor} w_{(k)}^2\right]=  \sum_{k=1}^{\lfloor(1-\rho)M_l \rfloor} \mathbb{E}\left[x_{(k)}^2\right],
\end{align}
where the expectation is taken over the DNN weights, $w_{(k)}$ represents the element with the $k$-th smallest absolute value in matrix $\mathbf{W}_l$ and $x_{(k)}\triangleq \vert w_{(k)} \vert$. Before calculating the expectations in (\ref{apb1}), we need a probability distribution that fits the weights in $\mathbf{W}_l$. According to \cite{isik}, the Laplacian distribution with zero mean can be a good fit for pretrained network weights. Hence, any $w_k$ in matrix $\mathbf{W}_l$ is assumed Laplacian distributed with its absolute value $x_k=\vert w_k \vert $ follows the exponential distribution with parameter $\lambda_l$, i.e., $x_k \sim \text{Exp}(\lambda_l)$. Moreover, based on the definition of $x_{(k)}$, the sequence $\{x_{(k)}\}_{k=1}^{M_l}$ is an ordered version of $\{x_k\}_{k=1}^{M_l}$. Building upon the theory of order statistics in \cite{order}, we rewrite $x_{(k)}$ into the statistically equivalent form
\begin{align}
    x_{(k)} = \frac{1}{\lambda_l} \sum_{j=1}^k \frac{z_j}{M_l-k+1},
\end{align}
where $z_j$ are independent and identically distributed standard exponential random variable. Then, we have
\begin{align}\label{apb2}
    \mathbb{E}\left[x_{(k)}^2\right]&=\mathbb{E}\left[ \left(\frac{1}{\lambda_l} \sum_{j=1}^k \frac{z_j}{M_l-k+1}\right)^2\right]\nonumber \\
    &=\frac{2}{\lambda_l^2}\sum_{j=1}^k \frac{1}{(M_l-j+1)^2}\nonumber \\
    &\quad+\frac{1}{\lambda_l^2}\sum_{j=1}^k\sum_{j^\prime \neq j}^k \frac{1}{(M_l-j+1)(M_l-j^\prime +1)}.
\end{align}
Plugging (\ref{apb2}) into (\ref{apb1}), we arrive at the exact value of $\mathbb{E}\left[ d_l\right]$. However, it still exhibits a complicated form and hinders further design. Note that $\mathbb{E}\left[ d_l\right]$ is only related to the pruning ratio $\rho$. Hence, we define the following function 
\begin{align}
    f(\rho)\!\triangleq\!\!\!\!\!\sum_{k=1}^{\lfloor(1\!-\!\rho)M_l \rfloor}\!\!\left(\sum_{j=1}^k\! \frac{1}{(M_l\!-\!j\!+\!1)^2}\!+\!\!\left( \sum_{j=1}^k\! \frac{1}{M_l\!-\!j\!+\!1}\right)^2\!\right)\!\!,
\end{align}
and $\mathbb{E}\left[ d_l\right]= \frac{1}{\lambda_l^2}f(\rho)$. Then, our goal is to find an accurate approximation with a concise form for $f(\rho)$. 

Defining $\Delta \rho \triangleq\frac{1}{M_l}$, we have 
\begin{align}\label{apb3}
    f(\rho+\Delta\rho)-f(\rho)=&-\sum_{j=1}^{\lfloor(1-\rho)M_l \rfloor} \frac{1}{(M_l-j+1)^2}\nonumber \\
    &-\left( \sum_{j=1}^{\lfloor(1-\rho)M_l \rfloor} \frac{1}{M_l-j+1}\right)^2.
\end{align}
For the first term on RHS of (\ref{apb3}), we have
\begin{align}
    \sum_{j=1}^{\lfloor(1-\rho)M_l \rfloor} \frac{1}{(M_l-j+1)^2}&\leq \sum_{j=1}^{\lfloor(1-\rho)M_l \rfloor} \frac{1}{(M_l-j+1)(M_l-j)}\nonumber \\
    &=\sum_{j=1}^{\lfloor(1-\rho)M_l \rfloor} \frac{1}{M_l-j}- \frac{1}{M_l-j+1}\nonumber \\
    &=\frac{1}{\lceil \rho M_l \rceil}-\frac{1}{M_l}.
\end{align}
As for the second term on RHS of (\ref{apb3}), we adopt the approximation in Eq. (0.131) of \cite{int} and get
\begin{align}
    \left( \sum_{j=1}^{\lfloor(1-\rho)M_l \rfloor} \frac{1}{M_l-j+1}\right)^2\approx \left(\log\left( \frac{M_l}{\rho M_l+1}\right)\right)^2.
\end{align}
Now, according to the above results and by definition,  we can calculate the derivative of $f(\rho)$ with respect to $\rho$ as
\begin{align}
    f^\prime (\rho)&=\lim_{\Delta\rho\to 0}\frac{f(\rho+\Delta\rho)-f(\rho)}{\Delta \rho }=\lim_{M_l \to \infty}\frac{f(\rho+\Delta\rho)-f(\rho)}{\Delta \rho }\nonumber \\
    &\approx \lim_{M_l \to \infty} - \frac{M_l}{\lceil \rho M_l \rceil}-M_l\left(\log\left( \frac{M_l}{\rho M_l+1}\right)\right)^2+1\nonumber \\
    &= \frac{\rho-1}{\rho}-M_l\log^2 \rho.
\end{align}
Moreover, noting that $f(1)=0$ and by performing indefinite integration, we have
\begin{align}\label{apb4}
    f(\rho)&\approx \!-2M_l \rho \!-\!M_l \rho \log^2\rho \!+\!2M_l \rho \log \rho \!+\!\rho\!-\!\log \rho \!+\!2M_l \!-\!1\nonumber \\
    &\approx M_l \left(2-\rho-\rho(\log\rho -1)^2\right),
\end{align}
where the second approximation holds when $M_l\gg 1$. 
Then, using its expectation, we can approximate $d_l$ by
\begin{align}\label{eq43}
    d_l \approx \frac{M_l}{\lambda^2} \left( 2-\rho -\rho\left( \log \rho -1\right)^2\right).
\end{align}
Due to the law of large numbers, as the number of parameters increases, the error in this approximation gradually decreases, leading to a more accurate approximation. Substituting (\ref{eq43}) into (\ref{upper}), we obtain (\ref{approx})
and complete the proof.


\section{Proof of Lemma \ref{lemma4}} \label{app3}
To tackle with the nonconvex objective, we introduce a new variable $t$ to replace $P_{\text{C}}$ and it follows $t=\frac{1}{\log_2\left(1+\frac{g P_\text{C}}{BN_0}\right)}$. Furthermore, we express $P_\text{C}$ as a function with respect to $t$, i.e., $P_{\text{C}}=\frac{BN_0}{g}\left(2^{\frac{1}{t}}-1\right)$, and the problem in (\ref{pro5}) is equivalently reformulated as
\begin{align}\label{proa1}
\mathop{\text{minimize}}_{t,\nu_{\text{e}}} \enspace& \frac{A_1 BN_0}{g}\left(2^{\frac{1}{t}}-1\right)t+\kappa A_2 \nu_{\text{e}}^2\nonumber \\
\text{subject to}\enspace& A_1 t+\frac{A_2}{\nu_{\text{e}}}\leq T_2,\nonumber\\
&t\geq t_{\min},\enspace \text{C}6,
\end{align}
where $t_{\min}\triangleq \frac{1}{\log_2\left(1+\frac{g P_{\max}}{BN_0}\right)}$. Now, the problem in (\ref{proa1}) exhibits a convex form and can be optimally solved by checking the KKT conditions. To this end, we express the Lagrangian function of (\ref{proa1}) as 
\begin{align}
    \mathcal{L}(t,\nu_{\text{e}})=& \frac{A_1 BN_0}{g}\left(2^{\frac{1}{t}}\!-\!1\right)t\!+\!\kappa A_2 \nu_{\text{e}}^2\!+\!\mu_1\left(A_1 t\!+\!\frac{A_2}{\nu_{\text{e}}}\!-\!T_2\right)\nonumber \\
    &+\mu_2(t_{\min}-t)+\mu_3(\nu_{\text{e}}-\nu_{\max}),
\end{align}
where $\mu_1$, $\mu_2$, and $\mu_3$ are Lagrange multipliers associated with the constraints in (\ref{proa1}), respectively. The KKT conditions can be listed as follows
\begin{align}\label{apc0}
    &\frac{\partial \mathcal{L}(t,\nu_{\text{e}})}{\partial t}\!=\! \frac{A_1 BN_0}{g}\left(2^{\frac{1}{t}}\!-\!\mathrm{ln}2 \frac{2^{\frac{1}{t}}}{t}\!-\!1\right)+\mu_1 A_1-\mu_2\!=\!0, \\ \label{apc46}
    &\frac{\partial \mathcal{L}(t,\nu_{\text{e}})}{\partial \nu_{\text{e}}}= 2\kappa A_2 \nu_{\text{e}}-\frac{\mu_1 A_2}{\nu_{\text{e}}^2}+\mu_3=0, \\ \label{apc47}
    &\mu_1\left(A_1 t\!+\!\frac{A_2}{\nu_{\text{e}}}\!-\!T_2\right)=0, \\ \label{apc48}
    &\mu_2(t_{\min}-t)=0, \\ \label{apc49}
    &\mu_3(\nu_{\text{e}}-\nu_{\max})=0, \\ \label{apc50}
    &\mu_1,\mu_2,\mu_3\geq 0.
\end{align}
We start with the equality in (\ref{apc49}). If $\mu_2>0$, we have $t=t_{\min}$. On the other hand, if $\mu_2=0$, the value of $t$ can be determined by equality (\ref{apc46}), which $\mu_1$ is an unknown constant. Combining the two cases, we arrive at the general expression in (\ref{eq25}). Similar to the above steps, by checking the equations in (\ref{apc47}) and (\ref{apc50}), we can express the optimal $\nu_{\text{e}}^*$ in (\ref{eq24}). To further determine the value of $\mu_1$, we substitute the optimal $t^*$ and $\nu_{\text{e}}^*$ into equation (\ref{apc48}). Since the RHS of (\ref{apc48}) is monotonic, we can find the unique $\mu_1$ via the bisection method. The proof completes.

\bibliographystyle{IEEEtran}
\bibliography{IEEEabrv,reference}

\end{document}